\documentclass[a4paper,twocolumn,accepted=2026-10-28]{quantumarticle}
\pdfoutput=1
\usepackage[utf8]{inputenc}
\usepackage[english]{babel}
\usepackage[T1]{fontenc}
\usepackage{amsmath}
\usepackage{hyperref}
\usepackage[numbers,sort&compress]{natbib}
\usepackage{tikz}
\usepackage{lipsum}
\usepackage{mathtools}
\usepackage{microtype}
\usepackage{amssymb}
\usepackage{mathrsfs}
\usepackage{bm, dsfont}

\usepackage{graphicx}
\usepackage{hypcap}
\usepackage{verbatim, float}
\usepackage{psfrag}
\usepackage{tikz}
\usetikzlibrary{arrows.meta, automata,
                positioning,
                quotes}
\usepackage[normalem]{ulem}
\usepackage{physics}
\usepackage{amsmath}
\usepackage{amsthm}
\usepackage{units}

\newcommand{\id}{\mathds{1}}

\newcommand{\cS}{\mathcal{S}}

\newcommand{\cW}{\mathcal{W}}

\newcommand{\be}{\begin{equation}}
\newcommand{\ee}{\end{equation}}

\newcommand{\chsh}{\mathcal{S}}

\newcommand{\BobL}{$\B^{\rm L}$}
\newcommand{\BobS}{$\B^{\rm S}$}
\newcommand{\A}{{\rm A}}
\newcommand{\B}{{\rm B}}

\renewcommand{\S}{\mathcal{S}}
\renewcommand{\sb}{\mathrm{s}}
\newcommand{\C}{\mathcal{W}}
\newcommand{\T}{\mathcal{T}}
\newcommand{\Cs}{\mathsf{W}}
\newcommand{\Ts}{\mathsf{T}}
\newcommand{\cs}{\mathsf{w}}
\newcommand{\ts}{\mathsf{t}}

\newcommand{\TTn}{\mathcal{T}_{n}}
\newcommand{\WWn}{\mathcal{W}_{n}}

\newtheorem*{result}{Result}

\newtheorem*{corollary}{Corollary}
\newtheorem{lemma}{Lemma}

\begin{document}

\title{Certification of quantum correlations at arbitrary distances through routed Bell tests}

\author{Pavel Sekatski{\textsuperscript{*}}}
\email{pavel.sekatski@gmail.com}
\affiliation{Department of Applied Physics, University of Geneva, Switzerland}
 \affiliation{Geneva Quantum Centre, University of Geneva, Switzerland}

\author{Jef Pauwels{\textsuperscript{*}}}
\email{jef.pauwels@unige.ch}
\affiliation{Department of Applied Physics, University of Geneva, Switzerland}
 \affiliation{Geneva Quantum Centre, University of Geneva, Switzerland}
\affiliation{Constructor University, 28759 Bremen, Germany}
\affiliation{Laboratoire d'Information Quantique, Universit\'e libre de Bruxelles (ULB), Belgium}
 
\author{Edwin Peter Lobo{\textsuperscript{*}}}
\email{edwin.lobo@ulb.be}
\affiliation{Laboratoire d'Information Quantique, Universit\'e libre de Bruxelles (ULB), Belgium}

\author{Stefano Pironio}
\affiliation{Laboratoire d'Information Quantique, Universit\'e libre de Bruxelles (ULB), Belgium}

\author{Nicolas Brunner}
\affiliation{Department of Applied Physics, University of Geneva, Switzerland}
 \affiliation{Geneva Quantum Centre, University of Geneva, Switzerland}

\maketitle
\begingroup
\renewcommand{\thefootnote}{*}
\footnotetext{{These authors contributed equally.}}
\endgroup

\begin{abstract}
Transmission loss represents a major obstacle to the device-independent certification of quantum correlations over long distances, limiting applications such as device-independent quantum key distribution (DIQKD). In this work, we investigate the recently proposed concept of routed Bell experiments, in which a particle sent to one side can be measured either near or far from the source. We prove that routed Bell tests involving only entangled qubits can certify quantum correlations even in the presence of arbitrary loss on the channel to the distant device. This is achieved by adapting concepts from self-testing and quantum steering to the routed Bell test framework. Finally, as a natural extension of our approach, we outline a DIQKD protocol that, in principle, is secure over arbitrary distances.
\end{abstract}

\section{Introduction}
Device-independent (DI) protocols aim to certify quantum properties—such as entanglement or security—without making assumptions about the internal functioning of the measurement devices~\cite{Acin2007}. At the core of any DI claim lies a fundamental classical-quantum gap that can be established solely from the observed correlations. In addition, such claims rely on a specific causal structure—typically that of a Bell experiment—in which spatially separated parties perform local measurements on a shared quantum state. Within this framework, the classical-quantum gap manifests as a violation of a Bell inequality and is synonymous with the notion of quantum nonlocality~\cite{Brunner2014}.

In practice, harnessing the power of quantum nonlocality comes with significant experimental challenges. 
In particular, local detection efficiencies must exceed a specific threshold; otherwise, a local model might exploit the so-called ``detection loophole'' to explain the observed correlations~\cite{Pearle1970,Clauser1974}.

The detection efficiency required to violate a Bell inequality varies considerably with the experimental setup. While Bell tests involving many measurement settings and high-dimensional entangled states can, in principle, tolerate very low efficiencies~\cite{Massar2002,Massar2003,Vertesi2010,Pal2015,Marton2023,Miklin2022,Xu2023}, practical implementations typically rely on low-dimensional systems such as qubits. In these scenarios, the efficiency thresholds remain relatively high—typically above 60\%—even when loss is the only imperfection~\cite{Brunner2014}. For applications such as device-independent quantum key distribution (DIQKD), the requirements are even more stringent, with standard protocols becoming insecure below 50\% efficiency~\cite{Acin2016,Masini2024}.

Although recent experiments have achieved loophole-free Bell violations~\cite{Hensen2015,Giustina2015} and proof-of-principle demonstrations of DIQKD~\cite{Nadlinger2022,Zhang2022}, distributing nonlocal correlations over long distances remains a major challenge. Transmission losses—which effectively reduce the overall detection efficiency—increase exponentially with distance, imposing severe limitations even when local detectors operate at unit efficiency. This has motivated the exploration of heralded entanglement generation~\cite{Simon2003,Gisin2010,Kolodynski2020,Steffinlongo2024}, though such approaches to DIQKD remain confined to cutting-edge experimental demonstrations~\cite{Lu2026device}.

\begin{figure}[t!]
    \centering
    \includegraphics[width=0.99\columnwidth]{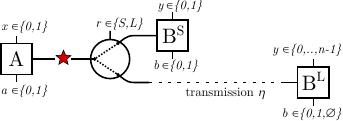}
    \caption{\emph{The routed Bell scenario}. A source distributes entanglement to Alice and Bob. A switch positioned after the source directs Bob's system either to a nearby measurement device $\B^{\rm S}$, or to a distant device $\B^{\rm L}$, which experiences loss with transmission $\eta$. We show that the observation of a Bell violation between $\A$ and $\B^{\rm S}$  allows the certification of genuine quantum correlations along the long path, for arbitrary transmission $\eta>0$. }
    \label{fig:1}
\end{figure}

Recently, an alternative approach to mitigating transmission loss in Bell experiments—coined ``routed Bell tests''—was proposed~\cite{Chaturvedi2022,Lobo2023}. In this scenario, illustrated in Fig.~\ref{fig:1}, a standard bipartite Bell experiment is modified so that one of the particles (e.g., Bob's) can be measured at two different locations: either near the source or at a distant site. Interestingly, these tests can tolerate more loss than standard Bell tests~\cite{Lobo2023}, and can increase robustness to loss for DIQKD~\cite{Deloison2024,Tan2024}. The setup is closely related to DI protocols with a ``local Bell test''~\cite{Lim2013}. 

In the present paper, we prove that routed Bell tests using entangled qubits can tolerate arbitrarily high transmission loss, enabling the distribution of genuine quantum correlations over an arbitrarily large distance. Our approach is analytical, drawing concepts from self-testing~\cite{Mayers2004, Bowles2018} and quantum steering~\cite{Wiseman2007,Uola2020,Skrzypczyk2015}. In particular, we derive a tight family of steering inequalities for qubits, which may be of independent interest. As an extension of our analysis, we outline a DIQKD protocol for the routed Bell scenario, which can in principle achieve a non-zero key rate at arbitrary distances.

\section{The routed Bell scenario}
In a routed Bell test~\cite{Chaturvedi2022,Lobo2023}, illustrated in Fig.~\ref{fig:1}, a source prepares a bipartite quantum state $\rho_{AB}$, sending subsystem $A$ to Alice—who is positioned close to the source—and subsystem $B$ towards Bob. Upon receiving an input $x$, Alice performs a measurement described by a set of positive operator-valued measures (POVMs) $\{\text{A}_{a|x}\}_x$ and obtains an outcome $a$. Bob's particle is routed by a switch to one of two measurement stations, depending on an input $r \in \{\rm S, \rm L\}$. If $r = {\rm S}$ (for ``short path'' or SP), the system is measured at a nearby device $\B^{\rm S}$; if $r = {\rm L}$ (for ``long path'' or LP), it is measured at a distant device $\B^{\rm L}$. The measurements at SP and LP are described by POVMs $\{\B^{\rm S}_{b|y}\}_y$ and $\{\B^{\rm L}_{b|y}\}_y$, respectively, where $y$ is the input and $b$ the output. These measurements may differ between the two locations. The correlations observed in the experiment are given by
\begin{equation} \label{eq:born}
    p(a,b|x,y,r) = 
        \tr (\rho_{AB}\A_{a|x} \otimes \B^r_{b|y}) \,.
\end{equation}
In the DI spirit, no assumptions are made about the internal workings of the source, the measurements, or the switch—all are treated as untrusted. The only physical assumptions are causality constraints enforcing the tensor product structure in Eq.~\eqref{eq:born}, and the independence of the devices. For example, Alice's measurement process is assumed to be independent of the inputs and outputs on Bob's side (which includes the switch and both measurement stations), and vice versa.

In the routed Bell scenario, the fact that the measurement devices $\A$ and $\B^{\rm S}$ are located close to the source reduces losses and noise, making it easier to observe a Bell violation in the SP. Such a violation allows for the partial certification of key properties of the source and Alice's measurement device. The central goal is then to certify genuine quantum correlations in the LP between $\A$ and $\B^{\rm L}$, possibly leveraging the information provided by the SP violation.

Unlike standard Bell tests and traditional device-independent (DI) protocols, our aim is not to exclude a classical explanation for the entire experiment—this is already achieved via the Bell violation in the SP. Instead, we seek to rule out a broader class of hybrid quantum-classical models, in which the devices near the source may operate quantumly, while only classical information reaches \BobL . Correlations of this kind are referred to as short-range quantum (SRQ), following Ref.~\cite{Lobo2023}.

SRQ correlations can thus be reproduced by strategies where the quantum system traveling along the LP is measured before reaching \BobL. More precisely, the process is modeled by a POVM $\{{{\rm E}_\lambda\}}_\lambda$ acting on Bob's quantum system, producing a classical variable $\lambda$. The distant device \BobL\ then outputs $b$ based on the input $y$ and $\lambda$, via a purely classical post-processing function $p(b|y,\lambda)$. SRQ correlations therefore admit the following decomposition:

\begin{align}\label{eq: def SRQ}
    p(a,b&|x,y,r) 
    \\
    &\!\stackrel{\rm SRQ}{=}
    \!\begin{cases}
        \tr (\rho_{AB} \A_{a|x} \otimes \B^{\rm S}_{b|y} ) &\!\! r\! =\! {\rm S},  \\
        \sum_\lambda p(b|y,\lambda) \tr(\rho_{AB} \A_{a|x} \otimes {\rm E}_\lambda ) &\!\!r \!= \!{\rm L}.
    \end{cases}\nonumber
\end{align}

Correlations that cannot be expressed in the above form are termed long-range quantum (LRQ). These require the distribution of entanglement along the LP and are essential for any application that relies on certifying quantum properties between distant parties~\cite{Lobo2023,Deloison2024}. Similar to nonlocal correlations in the standard Bell scenario, LRQ correlations can be detected by the violation of a ``routed'' Bell inequality~\cite{Lobo2023}. We now introduce a family of inequalities that can be violated by a quantum model even in the presence of arbitrary losses along the LP.

\section{A family of routed Bell inequalities}
We consider a routed Bell scenario where \A\ and \BobS\, have two binary-outcome measurements:  $a,x,b,y \in \{0,1\}$. The device \BobL\  has ternary outcomes
$b \in \{0,1,\varnothing\}$, where we will later identify the outcome $\varnothing$ with a no-click. The input of \BobL\ takes $n$ values $y\in\{0,\ldots,n-1\}$, that we also label for convenience as the `angles' $\theta_y = y {\pi}/{n}$. 

Our routed Bell inequality is based on the following three interconnected observable quantities: the SP CHSH expectation $\S$, 
\begin{equation}
    \label{eq: CHSH}
    \S = \sum_{x,y=0,1} (-1)^{x \cdot  y}\, \langle A_x B^{\rm S}_y\rangle \,,
\end{equation}
the LP asymmetric Bell expectation
\begin{equation}
\label{eq: C}
\C_n = \frac{1}{n}\sum_{y=0}^{n-1}\Big( \cos{\theta_y} \, \langle A_0 B_y^{\rm L}\rangle  + \sin\theta_y \, \langle A_1 B_y^{\rm L}\rangle \Big),
\end{equation}
and the average click probability of \BobL's detector
\begin{equation}
    \label{eq: T}
    \T_n = \frac{1}{n}\sum_{y=0}^{n-1} \sum_{b=0,1} p(b|y,r={\rm L}) \,.
 \end{equation}
 In the above expressions, we introduced the observables $A_x := \A_{0|x}-\A_{1|x}$ and $B_y^r := \B^r_{0|y}-\B^r_{1|y}$, defining the SP and LP correlators $\langle A_x B^r_y\rangle = \sum_{a,b = 0,1} (-1)^{a+b}\, p(a,b|x,y,r)$.  

As we will show, local correlations (in a standard Bell scenario where the SP test is ignored) satisfy the following non-linear Bell inequality $\C_n \leq \sqrt{2} \sin(\frac{\pi}{2}\T_n )/n\sin(\frac{\pi}{2 n})$.

In a routed Bell scenario, the SRQ bound on the LP quantity $\C_n$ is a function of the SP CHSH violation $\S$, and is reduced relative to the local bound by a factor that depends on $\S$. This is formalized in the following result.
\begin{result}[Routed Bell inequality]
For $\S\geq 2$, $\T_n>0$, and $\frac{\pi}{2}\T_n/\sin\left(\frac{\pi}{2}\T_n\right) \geq   \sqrt{ 2} \, \gamma(\S)$,
SRQ correlations satisfy the routed (non-linear) Bell inequality
   \be \label{eq:main}
   \C_n \leq   \frac{\sqrt{2}\, \sin\left(\frac{\pi}{2}\T_n\right)}{n\sin\left(\frac{\pi}{2 n}\right)} \gamma(\S)\,,
    \ee
    where $\gamma(\S) = \left(\S+\sqrt{8-\S^2}\right)/4$.
For $\T_n=0$, one has $\C_n=0$, and the bound holds trivially.
\end{result}
The SP factor $\gamma(\S)$ is equal to $1$ when $\S$ saturates the local bound $\S=2$ and decreases to $\sqrt{2}/2\simeq 0.71$ when $\S=2\sqrt{2}$.

We prove the above result in detail in App.~\ref{app: proof main} and sketch it below. There, we also generalize Eq.~\eqref{eq:main} to the whole range of values $\S$ and $\T_n$.  We also consider a second scenario where \BobL\ has a continuous input $y \equiv \theta \in [0,\pi)$, which can be thought of as the limit $n\to \infty$ of the discrete input case. We show that the above bound is tight in that limit.

Alternatively to Eq.~\eqref{eq:main}, we also prove that SRQ correlations satisfy a family of linear inequalities of the form 
$\C_n + \mu \T_n +\nu \S \leq \xi$, useful for experimental tests. Such linear inequalities can be readily combined with standard tools for finite-statistics analysis of the experimental data, enabling the computation of p-values and confidence intervals on the average violation.

\section{Sketch of the proof} 
The key idea underlying the proof is that the SP CHSH violation provides information about the measurements performed by Alice. For instance, when $\S=2\sqrt{2}$, Alice's measurements can be self-tested as two-qubit anti-commuting Pauli measurements. Since \A\ is characterized, the LP test can then be interpreted as a steering test in which \BobL\ is untrusted. Our proof therefore combines the partial self-testing of Alice's measurements from the SP with a steering inequality applied in the LP, treating \BobL\ as the untrusted party.

We begin by discussing the steering part of the argument. For this, we assume that Alice performs two qubit measurements parameterized by an angle $\alpha$ in the $X-Z$ plane. That is, we assume that Alice's observables 
are of the specific form
\begin{equation}
   A_x\equiv  A_x^\alpha = \cos\alpha \frac{Z+X}{\sqrt{2}} +(-1)^x \sin\alpha \frac{Z-X}{\sqrt{2}} \,.
\end{equation}
The correlations in the LP can then be written as $p(a,b|x,y,{\rm L}) = \tr_{B}(\sigma_{b|y}A_x^\alpha)$, where $\sigma_{b|y}$ is the assemblage created by the distant Bob through his measurement. If the correlations are SRQ, then the assemblage can be written as $\sigma_{b|y} = \sum_\lambda p(b|y,\lambda) \rho_\lambda$, where $\rho_\lambda = \tr_B(\rho_{AB} \id \otimes {\rm E}_\lambda)$, i.e., the assemblage admits, in the language of steering, an LHS model.
Thus, any steering inequality involving the trusted observables $A_x^\alpha$ is satisfied by the LP correlations for SRQ models.

We derive the following steering inequality in App.~\ref{app:steeringlemma}
\begin{equation}\label{eq:AJordan}
    \C_n^\alpha\leq \frac{\sqrt{2} \sin\left(\frac{\pi}{2}\T_n\right)}{n\sin\left(\frac{\pi}{2 n}\right)} \gamma\left(\sb_\alpha \right)\,,
\end{equation}
where we introduced $\sb_\alpha = 2(\cos\alpha + \sin\alpha)$.  That is, we will show that if the measurements of Alice are of the form $A_x^\alpha$, then the LP quantities $\WWn^\alpha$ and $\TTn$ satisfy Eq.~\eqref{eq:AJordan} whenever $\sigma_{b|y} = \tr_B(\rho_{AB}\id \otimes {\rm B}^{\rm L}_{b|y})$ admits an LHS model.
For complementary measurements $(\alpha={\pi}/{4})$ we derive a tight characterization of the LHS correlation set in terms of the quantities  $(\WWn^{\nicefrac{\pi}{4}},\TTn)$ in App.~\ref{app:steeringlemma}. These results are of independent interest and complement the steering inequalities of Ref.~\cite{Skrzypczyk2015}.

We now relax the assumption on $A_x$ and characterize the measurements of Alice based only on the SP CHSH violation.  Since Alice performs two binary-outcome measurements, Jordan's lemma~\cite{Jordan1875,Scarani2012} implies that her observables $A_x$ and the shared state $\rho_{AB}$ can be decomposed into qubit blocks as
\begin{equation}\label{eq:Jordan_decomposition}
    A_x =\bigoplus_i A_x^{\alpha_i}\,,\quad 
    \rho_{AB} = \bigoplus_i \mu_i \ketbra{\psi_{i}}{\psi_{i}} \,,
\end{equation}
where $\mu_i = \tr [\rho_{AB}\, \Pi^i_{A}\otimes \id_B]$ is the probability to find the state in the corresponding qubit block (see App.~\ref{app:CHSH}). 
Substituting Eq.~\eqref{eq:Jordan_decomposition} into Eq.~\eqref{eq: CHSH}, we get 
\begin{align}
            \chsh = 2 \sum_i \mu_i \Big(&\cos\alpha_i \bra{\psi_i}\frac{Z+X}{\sqrt{2}} \otimes B^{\rm S}_0 \ket{\psi_i} + \\   &\sin\alpha_i \bra{\psi_i}\frac{Z-X}{\sqrt{2}} \otimes  B^{\rm S}_1 \ket{\psi_i}\Big)
         \leq\sum_i \mu_i \sb_{\alpha_i} \nonumber
\end{align}
where we used the fact that $\|A \otimes B\| \le 1$ for any pair of dichotomic observables $A$ and $B$. Therefore, the distribution of the ``measurement angle'' $\alpha_i$  is constrained by the observed CHSH score 
\be \label{eq:SPtest}
\sum_i \mu_i\sb_{\alpha_i} \geq \chsh\,.
\ee
Further, from the steering inequality Eq.~\eqref{eq:AJordan}, we have
\begin{align}
\WWn = \sum_i \mu_i \WWn^{\alpha_i} \leq \frac{\sqrt{2}}{n\sin\left(\frac{\pi}{2 n}\right)}\, \sum_i \mu_i  \,\sin\left(\frac{\pi}{2}\TTn^i\right)\, \gamma(\sb_{\alpha_i}) \,.
\end{align}
As we show in App.~\ref{app: proof main}, the function $F(\sb_\alpha, \TTn) := \sin\left(\frac{\pi}{2}\TTn\right) \gamma(\sb_{\alpha})$ is equal to its concave envelope whenever $\frac{\pi}{2}\T_n/\sin\left(\frac{\pi}{2}\T_n\right) \geq  \sqrt{2} \gamma(\sb_\alpha)$. Therefore,
\begin{equation}
    \begin{aligned}
    n\sin\left(\frac{\pi}{2 n}\right) \WWn &\leq  \sqrt{2} \, \sum_i \mu_i \,F(\sb_{\alpha_i},\TTn^i)\\
    &\leq  \sqrt{2} \, F\left(\sum_i \mu_i \,\sb_{\alpha_i}, \sum_i \mu_i \, \TTn^i\right)\,\\
    &= \sqrt{2}  \, \sin\left(\frac{\pi}{2}\TTn\right) \gamma \qty(\sum_i \mu_i \sb_{\alpha_i})\,.
\end{aligned}
\end{equation}
Since the function $\gamma(x)$ is monotonically decreasing, we can bound $\gamma\left(\sum_i \mu_i \sb_{\alpha_i}\right)\leq \gamma(\S)$ from Eq.~\eqref{eq:SPtest}, which yields Eq.~\eqref{eq:main}.

We emphasize that the logic just presented provides a general method for promoting any family of steering inequalities---valid for Alice's trusted measurements \( A_x^\alpha \) across all angles \( \alpha \) (such as Eq.~\eqref{eq:AJordan})---into a routed Bell test. The formal mapping is provided in Lemma~\ref{app:steering_to_bell} of App.~\ref{app: proof main}.\\

We now show that a simple two-qubit strategy can yield a violation of our routed Bell inequality for an arbitrary transmission $\eta$ in the  LP. 
We consider a quantum implementation of the routed Bell experiment where the source prepares the two-qubit maximally entangled state $\ket{\Phi^+}_{AB} = (\ket{00} + \ket{11})/\sqrt{2}$, and the devices $\A$ and $\B^{\rm S}$ perform the Pauli measurements $\{Z,X\}$ and $\{\frac{Z\pm X}{\sqrt2}\}$, respectively, while the device $\B^{\rm L}$ performs $n$ measurements along the bases $\{Z\cos\theta_y + X \sin\theta_y \}_{y \in \{0,\ldots,n-1\}}$. 
Non-detection events $\varnothing$, corresponding to the situation when a detector fails to click, are treated as a separate outcome by $\B^{\rm L}$ and binned with one of the other two outcomes  by  $\A$ and $\B^{\rm S}$.

The correlation quantities can always be expressed as
\be
\S=2\sqrt{2}(1-\epsilon), \quad \T_n={\eta}, \quad \C_n=\eta(1-\delta)
\ee
where \( \eta \) represents the total detection efficiency of Bob's distant device $\B^{\rm L}$ and \( \epsilon \), \( \delta \) capture any other imperfections. Substituting these values into the routed Bell inequality \eqref{eq:main} gives us the following condition satisfied by all SRQ correlations
\begin{equation}\label{eq:condition_explicit}
    \begin{multlined}
     \frac{\eta}{\sin(\frac{\pi}{2}\eta)} \stackrel{\rm SRQ}{\leq} \frac{1-\epsilon + \sqrt{\epsilon(2-\epsilon)}}{1-\delta}\frac{1}{n \sin(\frac{\pi}{2 n})}.
    \end{multlined}
\end{equation}
For small imperfections and large number of settings, i.e., for $\epsilon$, $\delta$, and $1/n$ small, we can expand the right-hand side of the above inequality to obtain the following condition for demonstrating LRQ correlations with the qubit model:
\begin{equation}\label{eq:condition_small}
    \eta>\frac{1}{n}\sqrt{1+\frac{24}{\pi^2}n^2 \left(\delta +\sqrt{2\epsilon }\right)}\,.
\end{equation}
For ideal devices, only limited by the finite transmission $\eta$ in the LP, i.e., for $\epsilon=\delta=0$, LRQ correlations can be witnessed as long as $\eta>{1}/{n}$, that is, for arbitrarily low detection efficiencies provided that $\B^{\rm L}$ has enough measurement inputs.

 This shows that the routed Bell scenario offers a dramatic advantage in terms of robustness to loss compared to standard Bell tests. For comparison, we note that in a standard Bell test, the correlations $p(a,b|x,y,{\rm L})$ observed by Alice and $\B^{\rm L}$ in our example are local for $\eta\leq {2}/{\pi}\simeq 0.64$ (see App.~\ref{app: LHVmodels} for a detailed discussion). Additionally, it is worth noting that the required number of measurement settings $n >1/\eta$ is known to be minimal~\cite{Massar2003,Lobo2023}.

Inequality~\eqref{eq:condition_small} remains violated even when the local devices are sufficiently good, though not perfect, i.e., $\epsilon,\delta\ll 1$. To be concrete, consider a specific noise model, where all local detectors have limited efficiency $\eta_d$, the entanglement source has visibility $\nu$ (white noise), the transmission efficiency in the SP is $\eta_{SP}$, and the transmission efficiency in the LP is $\eta_{LP}$.
Assuming independent losses, and assigning no-click events at
Alice and $\B^{\rm S}$ to outcome $0$, the SP CHSH score is
\[
\S = 2\sqrt{2}\,\eta_d^2\eta_{SP}\nu
    +2(1-\eta_d)(1-\eta_d\eta_{SP}).
\]
Consequently,
\begin{equation}
\begin{split}
\eta &= \eta_{LP}\eta_d,
\qquad \delta = 1-\eta_d\nu,\\
\epsilon
    &= 1-\eta_d^2\eta_{SP}\nu
       -\frac{(1-\eta_d)(1-\eta_d\eta_{SP})}{\sqrt{2}}.
\end{split}
\nonumber
\end{equation}
Plugging these values into Eq.~\eqref{eq:condition_small}, it is easy to see that for any transmission efficiency $\eta_{LP}$ in the LP, LRQ correlations can be witnessed provided the imperfections of the local devices are sufficiently small. For instance, taking $\eta_d=\nu=\eta_{SP}=1-\zeta$, we find that the condition \eqref{eq:condition_small} is equivalent to
\begin{equation}
    \eta_{LP}>\frac{1}{n}(1+\frac{24}{\pi^2}n^2 \sqrt{2\zeta})\,.
\end{equation}
to lowest order in $\zeta$. The LP transmission efficiency $\eta_{LP}$ can thus be arbitrarily small provided $n$ is large enough and the local imperfections are of order $\zeta< 1/n^4$.

\section{Possibility of DIQKD with arbitrary loss} 
LRQ correlations are necessary for the security of DIQKD in routed setups~\cite{Deloison2024}. Having demonstrated that LRQ correlations can be certified over arbitrary distances, a natural question is whether such correlations also enable DIQKD. This is far from obvious since the requirements on the detection efficiency for DIQKD are typically more stringent than for certifying the distribution of entanglement~\cite{Acin2016}. Furthermore, the correlation scenario we introduced involves two measurement bases on Alice’s side. To maximize correlations for key generation, it would be natural to pair these two measurements with two corresponding measurements among Bob's $n$ possible settings. However, it is known that no DIQKD protocol based solely on two measurement settings for key generation can achieve security when the detection efficiency falls below $1/3$.

Nevertheless, our correlation scenario can be generalized by introducing multiple measurement settings for Alice and the SP device \BobS. Specifically, each party can implement $n$ measurement settings, employing the chained Bell inequality~\cite{Braunstein1989} rather than the CHSH inequality for self-testing. The underlying logic is analogous to the two-measurement case we have studied: the SP test using the chained Bell inequality certifies the state and Alice's measurements, which can then be leveraged for certifying LRQ correlations involving Bob's distant device.

Specifically, consider the setup of Fig.~\ref{fig:1}, but now each party (Alice, \BobS\, and \BobL) has $n$ possible binary measurements, described by the observables
\begin{equation}\label{MeasChained}
A_{x}=\cos(\frac{ x \pi}{n})Z+\sin(\frac{ x \pi}{n}) X
\end{equation}
with $x\in\{0,\ldots,n-1\}$. The short-path and long-path observables are, respectively,
\[
\begin{aligned}
B_y^{\rm S}
  &=\cos\left(\frac{(y+\frac12)\pi}{n}\right)Z
   +\sin\left(\frac{(y+\frac12)\pi}{n}\right)X,\\
B_y^{\rm L}
  &=\cos\left(\frac{y\pi}{n}\right)Z
   +\sin\left(\frac{y\pi}{n}\right)X,
\end{aligned}
\]
where $y\in\{0,\ldots,n-1\}$. We use the chained Bell expression
\[
\mathcal I_n:=\sum_{y=0}^{n-1}
\left\langle(A_y+A_{y+1})B_y^{\rm S}\right\rangle,
\qquad A_n:=-A_0.
\]
For $\ket{\Phi^+}$, these observables give
$\mathcal I_n=2n\cos(\pi/(2n))$.
For $n=2$, this inequality recovers Eq.~\eqref{eq: CHSH}
up to a flip in the outcome of Bob's second measurement.
As before, the source produces a maximally entangled two-qubit state $\ket{\Phi^+}$. Under these conditions, the SP correlations achieve the maximal violation of the $n$-input chained Bell inequality~\cite{Braunstein1989}. Using known self-testing results~\cite{Supic2016}, we certify that the state is $\ket{\Phi^+}$ and that Alice's observables are of the form \eqref{MeasChained}.  Consequently, conditioned on Alice's measurement outcome $a$ for setting $x$, the state transmitted to \BobL\, is the pure qubit state
\begin{equation}\label{eq:boxalice}
\rho_{a|x}=\frac{1}{2}\Bigg(\id +(-1)^a\Big(\cos\left(\frac{x \pi}{n}\right)Z+\sin\left(\frac{x \pi}{n}\right)X\Big)\Bigg).
\end{equation}

Hence, due to the test performed in the SP, the combination of the source and Alice's measurement device can be viewed as preparing states of the form \eqref{eq:boxalice}, which are then sent along the LP to \BobL. The LP correlations now correspond to a prepare-and-measure scenario in which \BobL\, performs qubit measurements of the form $B_{y}^{\rm L}=\cos(y \pi/n)Z+\sin(y \pi/n)X$. This recovers precisely the setup of the receiver-device-independent (RDI) QKD protocol introduced in Ref.~\cite{Ioannou2022,Ioannou_2022}. The parties can then securely extract a secret key based on Bob's input, and this is possible for any $\eta>1/n$. We have thus introduced a proof-of-principle routed DIQKD protocol capable of producing a non-zero key rate at arbitrarily low transmission, provided that the number of inputs $n$ is chosen sufficiently large.

Security is established here only in the ideal case, assuming perfect violation of the chained Bell inequality and exact self-testing. The first step towards studying the robustness of such a DIQKD protocol under realistic conditions and imperfections would be to analyze the robustness of the underlying chained Bell correlations, following the approach we presented here for the CHSH-based SP test.

\section{Conclusion} We have shown that quantum correlations can be certified at arbitrary transmission efficiencies in a routed Bell scenario. This can be achieved with a two-qubit source and without complicated heralding setups. A key insight underlying this result is that the short-path test functions as a (partial) self-test of Alice's device, allowing us to derive an explicit routed Bell inequality which imposes strong constraints on any potential classical simulation of the long-path correlations. Additionally, we presented a proof-of-principle DIQKD protocol tailored to the routed Bell test scenario which can tolerate arbitrary loss.  The protocol addresses the problem of losses for direct-link DIQKD, and highlights the possibility of a fully photonic approach, alternative to event-ready setups~\cite{Nadlinger2022,Zhang2022,Lu2026device}. We emphasize that, in the present work, we only prove that the asymptotic key rate is nonzero under the stated assumptions (via the reduction to Ref.~\cite{Ioannou2022}); providing a full practical security proof with SP imperfections and a finite bound on the key rate is one of the main theoretical challenges going forward. These results demonstrate the clear potential of the routed Bell scenario for long-distance quantum correlations and its applications.

In the future, it would be interesting to characterize the admissible trade-off between noise and loss in routed Bell tests, and develop efficient protocols for demonstrating nonlocal quantum correlations and their applications. In particular, we believe that considering the full statistics of the experiments on the long path could improve the noise versus loss threshold.

\begin{acknowledgements} 

We thank Sadra Boreiri, Nicolas Gisin, Renato Renner, Mirjam Weilenmann and Ramona Wolf for exciting discussions at an early stage of this project. P.S., J.P. and N.B. acknowledge financial support from the Swiss State Secretariat for Education, Research and Innovation (SERI) under contract number UeM019-3, and the Swiss National Science Foundation (NCCR-SwissMAP), and a grant from the Foundation for the University of Geneva (Funige). S.P. acknowledges funding from the VERIqTAS project within the QuantERA II Programme that has received funding from the European Union's Horizon 2020 research and innovation program under Grant Agreement No. 101017733 and the F.R.S-FNRS Pint-Multi program under
Grant Agreement No. R.8014.21, from the European Union's Horizon Europe research and innovation program under the project ``Quantum Security Networks Partnership'' (QSNP, Grant Agreement No. 101114043), from the F.R.S-FNRS through the PDR T.0171.22, from the FWO
and F.R.S.-FNRS under the Excellence of Science (EOS) program Project No. 40007526, from the FWO through the BeQuNet SBO project S008323N, from the Belgian Federal Science Policy through the contract RT/22/BE-QCI and the EU ``BE-QCI'' program. S.P. is a Research
Director of the Fonds de la Recherche Scientifique -FNRS. E.P.L. acknowledges support from the Fonds de la Recherche Scientifique - FNRS through a FRIA grant.
\end{acknowledgements}

\section*{Author contributions}
All authors contributed to the derivation of the results and to the writing of the manuscript. AI assistants were used for proofreading the final version (typos, grammar, equations); they were not used to generate scientific content.

\bibliographystyle{quantum}
\bibliography{references}

\onecolumn
\appendix
\section*{Appendix}
\setcounter{equation}{0}
\numberwithin{equation}{section}
\renewcommand{\theequation}{\thesection\arabic{equation}}
\section{Local hidden variable models for lossy singlet correlations}
\label{app: LHVmodels}

Here, we review some results on the critical detection efficiency for standard Bell nonlocality for the specific scenarios considered in this work, i.e., focusing on cases with asymmetric losses, where Alice's device operates at unit efficiency and Bob's has finite efficiency, $\eta$. This corresponds to a situation where transmission is the only source of loss, and Alice is close to the source while Bob is far away.

As the local dimension $d$ of the entangled state increases, nonlocality can be observed at arbitrarily low transmission rates, specifically $\eta \geq 1/d$ 
 using a Bell inequality with $d$ settings~\cite{Vertesi2010}. However, this result is largely impractical due to the increase in settings and Hilbert space dimension as $\eta$ decreases. Limiting the dimension or number of measurements thus provides more realistic insights.

General bounds on the detection efficiency can also be established based solely on the number of measurement settings. Specifically, when Bob's settings $n_B$ satisfy  $n_B \leq 1/\eta$, the correlations become local and can be reproduced by a local hidden-variable (LHV) model via data rejection~\cite{Pearle1970,Clauser1974}. In the scenarios considered in the main text—where Alice has two settings with binary outcomes and Bob has any number of settings with ternary outcomes (the third outcome corresponding to a no-click event)—all facets of the local polytope correspond to relabelings of the CHSH inequality~\cite{Pironio2014}. Therefore, demonstrating nonlocality in such cases reduces to violating a CHSH inequality. Consequently, all correlations in these scenarios become local when  $\eta \leq \eta_{CHSH} = 1/2$~\cite{Eberhard1993,Cabello2007}.

Another interesting case is when the number of settings is unrestricted but we fix the state to be a maximally entangled two-qubit state. For projective measurements, an LHV model can simulate correlations when $\eta \leq \eta_{PVM} = 1/2$~\cite{Gisin1999}. Below, we make two small extensions to this result: (1) When Alice's measurements are restricted to a plane, an LHV model can simulate correlations up to $\eta_{RPVM} = 2/\pi$; (2) for arbitrary POVMs by Bob, critical detection efficiency further decreases to $\eta_{POVM} = 1/4$ and $\eta_{RPOVM} = 1/\pi$, respectively. These results are summarized in the table below.

\begin{table}[H]
\centering
\begin{tabular}{|r || c |c |}
\hline
&Bob measures all PVMs& Bob measures all POVMs\\
\hline
\hline
Alice measures all PVMs &$\eta_*^{PVM}=\frac{1}{2}$&$\eta_*^{PVM}=\frac{1}{4}$\\
\hline
Alice measures all PVMs in a plane & $\eta_*^{RPVM}=\frac{2}{\pi}$ & $\eta_*^{RPOVM}=\frac{1}{\pi}$\\
\hline
\end{tabular}
\caption{The critical detection efficiencies (for Bob) below which the correlations obtained with a maximally entangled two-qubit state and the measurements specified in the table can be reproduced by an LHV model. \label{tab:Bell} }
\end{table}

Beyond the case of maximally entangled two-qubit states with projective measurements on Alice's side, we are unaware of any results demonstrating the locality of such lossy quantum correlations. It is worth noting that these correlations remain extremal in the set of probabilities, making it particularly challenging to construct LHV models. Specifically, convex decomposition methods typically used for white noise~\cite{Hirsch2017} cannot be applied here.

\subsection{LHV models for the singlet}
We start by summarizing the LHV model constructed in Ref.~\cite{Gisin1999}, before extending it to include (i) POVMs on Bob's side and (ii) a restriction to measurements in the plane.

\paragraph{The LHV model of  Gisin-Gisin --} 
First, we recall the model of Ref.~\cite{Gisin1999}, which considers the correlations obtained from the single state $\ket{\Psi^-}=\frac{1}{\sqrt{2}}(\ket{01}-\ket{10})$, with Alice and Bob performing measurements with inputs given by vectors $\bm x,\bm y$ on the Bloch sphere. Alice's measurements are ideal PVMs, with output $a=\pm 1$, while Bob's have a finite efficiency $\eta$, with output $b= +1,-1,\varnothing$.  The correlations observed in this scenario are given by
\begin{align}
p(a,b|\bm x,\bm y ) = \bra{\Psi^-} \text{A}_{a|\bm x} \otimes \text{B}_{b|\bm y} \ket{\Psi^-}\\
\end{align}
where ${\rm A}_{a|\bm x}=\frac{1}{2}(\id + a A_{\bm x})$,  ${\rm B}_{b|\bm y}=\begin{cases}
\frac{\eta}{2}(\id + b B_{\bm y}) & b=\pm 1\\
(1-\eta)\id & b=\varnothing\end{cases}$ and we introduced $A_{\bm x} = \bm x \cdot \bm \sigma$ and similarly for Bob.

Note that the choice of the maximally entangled states plays no role, as they are all related by a local unitary transformation, which can be absorbed as a relabeling of the measurement settings (since we consider an invariant set of possible measurements). In the model of~\cite{Gisin1999} Alice and Bob share a random anti-parallel pair of vectors on the sphere $(\bm \lambda,\bm \lambda'=-\bm \lambda)$. Alice response function is deterministic
\be
a(\bm \lambda,\bm x) = {\rm sign} \, \bm x \cdot \bm \lambda.
\ee
Bob's response function $p(b|\bm \lambda,\bm y)$ is slightly more complicated, he outputs no-click with probability 
\be
p(\varnothing|\bm \lambda', \bm y) = 1- |\bm y \cdot \bm \lambda' |
\ee
and otherwise outputs the sign $b= {\rm sign}\,  \bm y \cdot \bm \lambda'$. 
Straightforward algebra allows one to verify~\cite{Gisin1999} that this response function reproduces the singlet correlations $p(a,b|\bm x,\bm y )$ for $\eta = \eta_*^{PVM} =\frac{1}{2}$. It is straightforward to adjust the model for any lower value of transmission.\\

\paragraph{Modification of the Gisin-Gisin model to include POVMs on Bob's side--}
\label{app: POVM model}
We now show how the above model can be modified to also simulate lossy POVM measurements on Bob's side. As the first step, note that by relabeling $b=1\to b=0$ and merging $b=-1,\varnothing \to b=\varnothing$ the above model can be used to simulate any two-outcome POVM with $\eta=\frac{1}{2}$
\be
{\rm B}_{b|\bm y} = \begin{cases} \frac{1}{2} \ketbra{\bm y} & b=0 \\
\id -\frac{1}{2} \ketbra{\bm y} & b=\varnothing,
\end{cases}  
\ee
where $ \ketbra{\bm y}  = \frac{1}{2}(\id + \bm y \cdot \bm \sigma)= \frac{1}{2}(\id + B_{\bm y}).$\\

Now let us consider any extremal POVMs on Bob's side, which in the ideal case reads $\{B'_{b}= \alpha_b \ketbra{\bm y_b}\}$ with $\sum_{b=0}^{n-1}\alpha_b = 2$ and $\sum_b \alpha_b\,  {\bm y}_b = 0$, and has $2\leq n\leq 4$ outcomes~\cite{Dariano2005}. After adding losses of transmission $\eta'$ before this measurement (or lowering its efficiency to $\eta'$ ) the POVM becomes
\be
{\rm B}_b = 
\begin{cases}
    \eta' \, \alpha_b \ketbra{\bm y_b} & 0\leq b \leq n-1 \\
    (1-\eta')\id & b=\varnothing 
\end{cases}.
\ee
For $\eta'=\frac{1}{4}$ this POVM can be decomposed as a convex mixture of $n$ two-outcome POVMs, labeled by $b'\in\{0,\dots,n-1\}$,
\be
{\rm B}_{b|b'} = 
\begin{cases}
\frac{1}{2} \ketbra{\bm y_b} & b=b'\\
\id -\frac{1}{2} \ketbra{\bm y_b}& b=\varnothing\\
    0 & \text{otherwise},
\end{cases}
\ee
Indeed for $p_b'=\frac{\alpha_{b'}}{2}$ (with $\sum p_{b'} =1$) one verifies that
\be
{\rm B}_b =\sum_{b'} p_{b'} {\rm B}_{b|b'} = 
\begin{cases}
    \frac{1}{4} \, \alpha_b \ketbra{\bm y_b} & 0\leq b \leq n-1 \\
    (1-\frac{1}{4})\id & b=\varnothing .
\end{cases}
\ee
Hence, for $\eta \leq \eta_{*}^{POVM}=\frac{1}{4}$ the correlations $p(a,b|\bm x,\{{\rm B}_b\})$ obtained with a maximally entangled states projective measurements of Alice and any lossy measurement of Bob admit a LHV representation.\\

\paragraph{ A LHV model for correlations in a plane. --} 
Now let us restrict the set of possible measurements of Alice to real projective measurements, that is we now have ${\rm A}_{a|\varphi} = \frac{1}{2}(\id +a A_\varphi)$ with
\be
A_\varphi = \cos \varphi\, Z + \sin \varphi \, X.
\ee
We will also change the state to $\ket{\Phi^+}=\frac{1}{\sqrt 2} (\ket{00}+\ket{11})$.
Since $A_\varphi$ and the shared $\Phi^+$ are real (in the computational basis) we find 
\[
p(a,b|\bm x,\bm y ) = \bra{\Phi^+} \text{A}_{a|\bm x} \otimes \text{B}_{b|\bm y} \ket{\Phi^+} = \bra{\Phi^+} \text{A}_{a|\bm x} \otimes \text{B}_{b|\bm y}^* \ket{\Phi^+} = \bra{\Phi^+} \text{A}_{a|\bm x} \otimes \frac{\text{B}_{b|\bm y}+\text{B}_{b|\bm y}^*}{2} \ket{\Phi^+}.
\]
i.e. the only contribution to the correlations comes from the real part ${\rm Re}[ \text{B}_{b|\bm y}] = \frac{\text{B}_{b|\bm y}+\text{B}_{b|\bm y}^*}{2}$ of Bob's measurement operators. We can thus limit his measurements to be real without loss of generality.\\

First, let us assume that they are also projective ${\rm B}_{b|\bm y}= \frac{1}{2}(\id +  b \bm y \cdot \bm \sigma)$. The real part of such PVM 
${\rm Re}[ \text{B}_{b|\bm y}]$ is a convex combination of projective measurements in the plane
${\rm B}_{b|y}= \frac{1}{2}(\id + b B_{y})$ with $B_{y} =\cos \theta \, Z + \sin \theta \, X$. Hence, it is sufficient to consider Bob measuring ${\rm B}_{b|\theta}$. The correlations are then given by 
\begin{align}\label{app: correlations real}
p(a,b|\varphi,\theta ) &= \eta \bra{\Phi^+} \frac{1}{4}\left(\id + a A_\varphi + b B_\theta + ab A_\varphi\otimes B_\theta \right) \ket{\Phi^+} \\
& =  \frac{\eta}{4}\left(1 + a\,b \cos(\varphi-\theta) \right).
\end{align}
We now consider the following LHV model. Let Alice and Bob share a random vector $\bm \lambda =\binom{\cos \zeta}{\sin \zeta}$ on the circle $\zeta\in[0,2\pi]$. As before Alice's response function is deterministic 
\be
a(\bm \lambda, \bm x) = a(\zeta, \varphi) = {\rm sign} \, \bm x \cdot \bm \lambda= \text{sign} \cos(\zeta-\varphi),
\ee
which gives the right marginals $p(a|\varphi)=\frac{1}{2}$.
Bob's response function is similar as before but with a slightly different  
dependence, he outputs no-click with the probability 
\be
p(\varnothing|\zeta, \theta) = 1-|\cos(\zeta-\theta)|,
\ee
and otherwise outputs $b ={\rm sign}\,\cos(\zeta-\theta)$. For the overall no-click probability we find
\be
\eta_{*}^{RPVM}= p(b\neq \varnothing|\theta)= \int_{0}^{2\pi} \frac{\dd \zeta}{2\pi} |\cos(\zeta-\theta)| = \frac{2}{\pi}.
\ee
To see that the full correlations are correct, we introduce the indicator function
\be
\chi(\alpha) =\begin{cases}
    1 & \cos(\alpha) \geq 0 \\
    0 & \text{otherwise}
\end{cases},
\ee
and 
compute
\begin{align}
p(a=1,b=1|\varphi,\theta) &= \int\frac{\dd \zeta}{2\pi} |\cos(\zeta-\theta)| \chi(\zeta-\varphi) \chi(\zeta-\theta)  \\
&= \int \frac{\dd \zeta}{2\pi} |\cos(\zeta)| \chi(\zeta+\theta-\varphi) \chi(\zeta) \\
& =\int_{-\pi/2}^{\pi/2} \frac{\dd \zeta}{2\pi} \cos(\zeta) \chi(\zeta+\theta-\varphi) 
\\
& = \frac{1+\cos(\theta-\varphi)}{2\pi} = \eta_{*}^{RPVM} \frac{1+\cos(\theta-\varphi)}{4}.
\end{align}
This is indeed the same value as in Eq.~\eqref{app: correlations real}. The computation for the remaining probabilities can be done similarly, leading to the conclusion that the LHV model simulates the lossy correlations obtained with a maximally entangled state, Alice performing projective measurements restricted to a plane, and Bob performing any projective measurements with efficiency $\eta\leq \eta_{*}^{RPVM}=\frac{2}{\pi}.
$
\\

Finally, by replicating the discussion of Section.~\ref{app: POVM model} it is straightforward to show that the model can be modified to work when Bob performs any POVM  with efficiency $\eta\leq \eta_{*}^{RPOVM}=\frac{1}{\pi}$.\\

\section{Steering inequalities} \label{app:steeringlemma}

In this appendix, we prove the steering inequality discussed in the main text. In the next section~\ref{app: stearing from main} we demonstrate the inequality~\eqref{eq:AJordan} from the main text, relying on a steering inequality demonstrated later (Sec.~\ref{app: steering n}).  In the two following sections \ref{app: steering n} and \ref{app: steering cont}, we present the full characterization of values $(\C_n,\T_n)$ admissible by LHS models when Alice performs two complementary Pauli measurements. Section \ref{app: steering n} treats the case of finite $n$, while section \ref{app: steering cont} discusses the case where \BobL\, has a continuous input.

\subsection{Proof of the steering inequality~\eqref{eq:AJordan}}

\label{app: stearing from main}

\begin{lemma}[A family of steering inequalities]
All correlations $p(a,b|x,y)$ unsteerable for Alice's trusted measurement operators $  A_x^\alpha = \cos\alpha \frac{Z+X}{\sqrt{2}} +(-1)^x \sin\alpha \frac{Z-X}{\sqrt{2}}$, i.e., admitting a LHS model 
\begin{align}
 p(a,b|x,y)  \stackrel{\rm LHS}{=}   \sum_\lambda  \tr \big( \A_{a|x}^\alpha \, \rho_\lambda \big) p(b|y,\lambda) \quad \text{with} \quad \rho_\lambda \succeq 0,\, \tr  \sum_\lambda \rho_\lambda = 1
\end{align}
and $\A_{a|x}^\alpha = \frac{1}{2}(\id +(-1)^a  A_x^\alpha)$, satisfy the steering inequality
\begin{equation}
n\sin\left(\frac{\pi}{2 n}\right)\C_n^\alpha \leq \sqrt{2} \sin\left(\frac{\pi}{2}\T_n\right) \gamma\left(\sb_\alpha \right)
\end{equation}
where ${\rm s_\alpha} = 2 (\cos(\alpha)+\sin(\alpha))$, $\gamma({\rm s}_\alpha) = \left({\rm s}_\alpha+\sqrt{8-{\rm s}_\alpha^2}\right)/4$, and 
\begin{align}
\C_n^\alpha &= \frac{1}{n}\sum_{y=0}^{n-1}\sum_{a,b=0}^1 (-1)^{a+b} \Big( \cos{\theta_y} \,  p(a,b|0,y) + \sin\theta_y \, p(a,b|1,y)\Big),
\\
    \T_n &= \frac{1}{n}\sum_{y=0}^{n-1} \sum_{b=0,1} p(b|y) \,.
 \end{align}
\end{lemma}

\begin{proof}
    We begin by rewriting the correlators for the LHS model as
    \begin{equation}\label{eq:Jordan_corr}
        \begin{aligned}
          \sum_{a,b=0}^1 (-1)^{a+b} p(a,b|x,y) = \tr \big(A^\alpha_x \sum_\lambda \rho_\lambda (p(0|y,\lambda) -p(1|y,\lambda) )\big)   
             =   \tr\big(A_x^\alpha(\rho_{0|y}- \rho_{1|y})\big),
        \end{aligned}
    \end{equation}
    where $\rho_{b|y}= \sum_\lambda p(b|y,\lambda)\rho_\lambda$.
For $\alpha\in[0,\pi/2]$, define
$\chi_\alpha:=\operatorname{sgn}(\cos\alpha-\sin\alpha)$,
with $\operatorname{sgn}(0)=0$. Since
$\sqrt{8-\sb_\alpha^2}=2|\cos\alpha-\sin\alpha|$,
we use $\chi_\alpha$ to keep track of the sign of $\cos\alpha-\sin\alpha$ below.

 This lets us write $\WWn^\alpha$ as
        \begin{align}
        \WWn^\alpha &= \frac{1}{n} \sum_{y=0}^{n-1} \tr \Big( (\cos\theta_y  A_0^\alpha + \sin\theta_y A_1^\alpha) (\rho_{0|y}- \rho_{1|y})\Big)\\
        &= \frac{\cos\alpha+\sin\alpha}{\sqrt{2}}\, \Cs_n + \frac{\cos\alpha-\sin\alpha}{\sqrt{2}}\, \overline{\Cs}_n \,\\
        &= \frac{\sb_\alpha}{2\sqrt{2}}\,\Cs_n + \chi_\alpha\frac{\sqrt{8-\sb_\alpha^2}}{2\sqrt{2}}\,\overline{\Cs}_n\,, \label{eq:W_steer}
    \end{align}
where
\begin{align}
    \Cs_n &:=\frac{1}{n}\sum_{y=0}^{n-1} \tr[ (Z\cos \theta_y+ X\sin \theta_y ) (\rho_{0|y}-\rho_{1|y})] \label{def:Wn} \,,\\
    \overline{\Cs}_n &:=\frac{1}{n}\sum_{y=0}^{n-1} \tr[ (Z\sin \theta_y+ X\cos \theta_y ) (\rho_{0|y}-\rho_{1|y})] \label{def:tilde_Sn}\,.
\end{align}
Similarly, $\TTn$ can be written as
\begin{align}
    \TTn = \Ts_n= \frac{1}{n} \sum_{y=0}^{n-1} \tr[(\rho_{0|y}+ \rho_{1|y})]\,.
\end{align}
We will show in sec.~\ref{app: steering n} that for any assemblage $\{\rho_{b|y}\}_{b,y}$ that admits a LHS model ($\rho_{b|y} =\sum_\lambda p(b|y,\lambda)\rho_\lambda$), the following inequalities hold:
\begin{equation}
|\Cs_n|,\ |\overline{\Cs}_n|
\leq \frac{\sin\left(\frac{\pi}{2}\Ts_n\right)}
{n\sin\left(\frac{\pi}{2n}\right)}.
\end{equation}
The bounds also hold for the negatives of these witnesses, since exchanging
Bob's outcomes $0$ and $1$ preserves the LHS property and the click probability.
Applying these absolute-value bounds to Eq.~\eqref{eq:W_steer}, and using
\[
\frac{\cos\alpha+\sin\alpha+|\cos\alpha-\sin\alpha|}{\sqrt{2}}
=\sqrt{2}\,\gamma(\sb_\alpha),
\]
gives Eq.~\eqref{eq:AJordan} for every $\alpha\in[0,\pi/2]$.

\end{proof}

\subsection{Steering inequality with discrete input $\theta_y$ }
\label{app: steering n}

In this section we consider the characterization of the set of unsteerable correlations (LHS) in terms of the following discrete steering observables  

\begin{align}
\label{eq: C disc}
    \Cs_n=\frac{1}{n}\sum_{y=0}^{n-1} \cs (\theta_y)& \qquad \text{with}\qquad \cs(\theta)= \tr[ (\cos(\theta) \,Z +\sin(\theta)\, X) (\rho_{0|\theta}-\rho_{1|\theta})]\\
     \label{eq: T disc}
    \Ts_n =\frac{1}{n}\sum_{y=0}^{n-1} \ts(\theta_y)& \qquad \text{with}\qquad \ts(\theta)= \tr (\rho_{0|\theta}+\rho_{1|\theta}),
\end{align}
where Bob's input is discrete and labeled by the angles $\theta_y = y\frac{\pi}{n}$ for $y \in \{0,\dots,n-1\}$.
We now prove the following result.\\

\begin{lemma}[LHS set for discrete setting] \label{thrm: steering disc}
    The set of values of $(\Ts_n,\Cs_n)$ in Eq.~(\ref{eq: T disc},\ref{eq: C disc}) attainable by LHS models is the polytope with vertices (extremal points) $\left(\widehat{\Ts}_{k|n},\pm \widehat{\Cs}_{k|n}\right)$ for $k=0,\dots, n$ and
    \begin{align}
      \widehat{\Ts}_{k|n} &= \frac{k}{n} \\
     \widehat{\Cs}_{k|n} &=  \frac{\sin \left( \frac{\pi}{2}  \frac{k}{n}\right)}{n \sin \left(\frac{\pi }{2 n}\right)}.
    \end{align}    
\end{lemma}

\begin{proof}
We want to find all pairs of values $(\Ts_n,\Cs_n)$ compatible with a LHS model. Recall that, since $\ts(\theta)$ and $\cs(\theta)$ only depend on the real part of the stats $\rho_{b|\theta}$, one can assume without loss of generality that the hidden states are sampled from $\{\mu_\zeta \Psi_\zeta \}$ where all $\Psi_\zeta$ in Eq.~\eqref{eq: pure real} are pure and real (see App.~\ref{app:steeringlemma} for the full argument).\\

Let us fix a specific hidden state $\Psi_\zeta=\frac{1}{2}(\id + \cos(\zeta) Z + \sin(\zeta) X)$ and consider the values $(\Ts_n^{(\zeta)},\Cs_n^{(\zeta)})$ that can be attained by varying the response function $p(b|\theta_y,\zeta)$. As before we have 
\begin{align}
   \ts^{(\zeta)} (\theta_y) &=  p(0|\theta_y,\zeta)+ p(1|\theta_y,\zeta)\\
    \cs^{(\zeta)}(\theta_y) &=  \big(p(0|\theta_y,\zeta) -p(1|\theta_y,\zeta)\big)\cos(\theta_y-\zeta),
\end{align}
leading to the  bounds 
\begin{align}
0 &\leq \ts^{(\zeta)}(\theta_y) \leq 1\\
|\cs^{(\zeta)}(\theta_y) | &\leq |\cos(\theta_y-\zeta)|\, \ts^{(\zeta)}(\theta_y),
\end{align}
that can be saturated by the right choice of the response function $p(b|\theta_y,\zeta)$. For the average quantity $\Cs_n^{(\zeta)} =\frac{1}{n}\sum_{y=0}^{n-1} \cs^{(\zeta)}(\theta_y)$ this gives the bound
\begin{align}
  -\frac{1}{n}\sum_{y=0}^{n-1} |\cos(\theta_y-\zeta)|\, \ts^{(\zeta)}(\theta_y)  \leq \Cs_n^{(\zeta)}\leq \frac{1}{n}\sum_{y=0}^{n-1} |\cos(\theta_y-\zeta)|\, \ts^{(\zeta)}(\theta_y)
\end{align}
which is also saturable.  The set containing the attainable values $(\Ts_n^{(\zeta)},\Cs_n^{(\zeta)})$ is thus described by the linear constraints 
\begin{align}
    \Ts_n^{(\zeta)} &=\frac{1}{n}\sum_{y=0}^{n-1} \ts^{(\zeta)}(\theta_y)\\
    |\Cs_n^{(\zeta)}| &\leq \frac{1}{n}\sum_{y=0}^{n-1} |\cos(\theta_y-\zeta)|\, \ts^{(\zeta)}(\theta_y) \\
    0 \leq \,&\ts^{(\zeta)}(\theta_y) \, \leq 1 \qquad \forall y,
\end{align}
and is a polytope, whose vertices are deterministic strategies (each $\ts^{(\zeta)}(\theta_y) $ is 0 or 1). However, not all deterministic strategies are vertices (some are inside the polytope), and we want to find those which are.  

For all vertices the value of $\Ts_n^{(\zeta)}$ is of the form $\widehat{\Ts}_{k|n}^{(\zeta)} = \frac{k}{n}$ for $k=0,\dots n$ (where $k$ is the number of $\ts^{(\zeta)}(\theta_y)=1$), and it remains to determine the corresponding extremal values of  $\widehat{\Cs}_{k|n}^{(\zeta)}$. We find 
\begin{equation}
|\widehat{\Cs}_{k|n}^{(\zeta)} | = \max_{\sum_y \ts^{(\zeta)}(\theta_y) = k} \frac{1}{n}\sum_{y=0}^{n-1} |\cos(\theta_y-\zeta)|\, \ts^{(\zeta)}(\theta_y)  = \frac{1}{n} \sum_{y \in \mathds{I}_k^{(\zeta)}} |\cos(\theta_y-\zeta)| 
\end{equation}
with  $\mathds{I}_k^{(\zeta)} \subset\{0,\dots,n-1\}$ denoting set of $k$ indices with the highest values of $|\cos(\theta_y-{\zeta})|$. Therefore, we conclude that for a fixed local hidden state $\Psi_\zeta$, the set of attainable values is contained in the polytope with vertices
\be\label{eq: ext points}
 \left(\widehat{\Ts}_{k|n}^{(\zeta)}  = \frac{k}{n},\, \pm\, \widehat{\Cs}_{k|n}^{(\zeta)}  =\frac{1}{n} \sum_{y \in \mathds{I}_k^{(\zeta)}} |\cos(\theta_y-\zeta)| \right) \qquad \text{for } k=0,\dots,n
 \ee\\

Next, we consider the values $(\Ts_n,\Cs_n)$ attainable by general LHS models, involving mixture of different $\zeta$. By convexity the extremal points of this set are in the union of the points $ \left( \widehat{\Ts}_{k|n}^{(\zeta)} , \pm\, \widehat{\Cs}_{k|n}^{(\zeta)}  \right)$ 
for all $\zeta$ and $k$. Nevertheless, since $\widehat{\Ts}_{k|n}^{(\zeta)}=\frac{k}{n}$ is independent of $\zeta$, it is sufficient to determine 
\begin{equation}
\widehat{\Cs}_{k|n} := \max_\zeta \, \widehat{\Cs}_{k|n}^{(\zeta)}
\end{equation}
in order to obtain the vertices of the whole  LHS set.

To solve this maximization, note that without loss of generality we can restrict $\zeta \in [0,\frac{1}{2} \frac{\pi}{ n}]$ by the symmetry of the problem. In this case, noting that $|\cos((n-y)\frac{\pi}{n}-\zeta)|= 
|\cos(y \frac{\pi}{n}+\zeta)|$ we see that the largest values of $|\cos(\theta_y-\zeta)|$ are attained in the decreasing order by $\theta_0, \theta_1, \theta_{n-1},\theta_2, \theta_{n-2},\dots$ so that
\begin{align}
\widehat{\Cs}_{k|n}^{(\zeta)} &= \frac{1}{n}\left(\cos(\zeta) + \cos(\frac{\pi}{n} - \zeta)+ \cos(\frac{\pi}{n} + \zeta) +
\cos(2 \frac{\pi}{n} - \zeta)+\cos(2 \frac{\pi}{n} + \zeta)+ \dots\right)
\end{align}
 with exactly $k$ terms. Here by construction all the angles are mapped to the  interval $[0,\frac{\pi}{2}]$ such that the cosines are all positive. It is then convenient to split the analysis in two branches, with even $k=2j$ and odd $k =2j+1$. For the two cases we obtain
\begin{align}\label{eq: sum even}
\widehat{\Cs}_{2j|n}^{(\zeta)} & =\frac{1}{n}\sum_{\ell=1}^j \left(\cos((\ell-1) \frac{\pi}{n}+\zeta)+\cos(\ell \frac{\pi}{n}-\zeta) \right)\\
 \label{eq: sum odd}
\widehat{\Cs}_{2j+1|n}^{(\zeta)} & =  \frac{1}{n}\cos(\zeta) + \frac{1}{n}\sum_{\ell=1}^j \left(\cos(\ell \frac{\pi}{n}-\zeta)+\cos(\ell \frac{\pi}{n}+\zeta) \right).
\end{align}
 Recall that in all of the above expressions the argument of the cosines are in the interval $[0,\frac{\pi}{2}]$, where $\cos(x)$ is a concave function satisfying $\cos(x_1)+\cos(x_2)\leq 2\cos(\frac{x_1+x_2}{2})$. Therefore, we have the following inequalities
\begin{align}
&
\cos(\ell \frac{\pi}{n}- \frac{\pi}{n}+\zeta)+\cos(\ell \frac{\pi}{n}-\zeta) \leq 2\cos(\ell \frac{\pi}{n}-\frac{\pi}{2 n})\\
&\cos(\ell \frac{\pi}{n}-\zeta)+\cos(\ell \frac{\pi}{n}+\zeta) \leq 2\cos(\ell \frac{\pi}{n}) 
\end{align}
and also $\cos(\zeta)\leq 1$. By  using these inequalities for all term in the sums of Eqs.~(\ref{eq: sum even},\ref{eq: sum odd}) we obtain
\begin{align}
\widehat{\Cs}_{2j|n}^{(\zeta)} &\leq  
\widehat{\Cs}_{2j|n}^{(\frac{\pi}{2n})}  =\frac{2}{n}\sum_{\ell=1}^j \cos((\ell-\frac{1}{2}) \frac{\pi}{n}) =\frac{1}{n \sin \left(\frac{\pi }{2 n}\right)}  \sin \left( \pi  \frac{2j}{2n}\right)
\\
\widehat{\Cs}_{2j+1|n}^{(\zeta)} & \leq \widehat{\Cs}_{2j+1|n}^{(0)}  = \frac{1}{n} + \frac{2}{n}\sum_{\ell=1}^j \cos(\ell \frac{\pi}{n})
=
\frac{1}{n \sin \left(\frac{\pi }{2 n}\right)}   \sin \left(\pi \frac{  2j+1}{ 2n}\right),
\end{align}
telling us that that one of the two values $\zeta=0, \frac{\pi}{2 n}$ maximizes $\widehat{\Cs}_{k|n}^{(\zeta)}$ depending on the parity of $k$. 
Combining the two expressions gives the maximal value for all $k$, which takes care of the optimization over $\zeta$
\be
\widehat{\Cs}_{k|n} =  \frac{\sin \left( \pi  \frac{k}{2n}\right)}{n \sin \left(\frac{\pi }{2 n}\right)}
\ee
We have thus established that the pairs of values $(\widehat{\Ts}_{k|n}= \frac{k}{n}, \pm \widehat{\Cs}_{k|n})$ for $k=0,\dots,n$ are saturable and give the LHS set. 

Let us also verify that all of these points are extremal, i.e. none of them are inside the polytope. This can be done by showing that the slope of the line connecting each subsequent pair of points is strictly decreasing
\be
\widehat{\Cs}_{k+1|n} -\widehat{\Cs}_{k|n}>
\widehat{\Cs}_{k+2|n} -\widehat{\Cs}_{k+1|n}\quad\Longleftrightarrow \quad \sin \left( \pi  \frac{k+1}{2n}\right) - \sin \left( \pi  \frac{k}{2n}\right) > \sin \left( \pi  \frac{k+2}{2n}\right)- \sin \left( \pi  \frac{k+1}{2n}\right)
\ee
for all $k=0,\dots,n-2$. The last inequality can be rewritten as
\be
\int^{\pi  \frac{k+1}{2n}}_{\pi  \frac{k}{2n}} \cos(x) \, \dd x > \int^{\pi  \frac{k+2}{2n}}_{\pi  \frac{k+1}{2n}} \cos(x) \,\dd x
\ee
which is easy to see from the fact that $\cos(x)$ is a strictly decreasing function for $x\in\left(0,\frac{\pi}{2}\right]$.\\
\end{proof}

The Lemma \ref{thrm: steering disc} gives a tight characterization of the LHS set in terms of its extremal points, or in other words its facets (piece-wise linear function of $\Ts_n$). However, it will also be useful to give a relaxation of the boundary in terms of a differentiable function of $\Ts_n$.\\

\begin{corollary}[Steering for discrete setting] \label{lemma: disc bound}
     For any LHS model the values of $(\Ts_n,\Cs_n)$ in Eq.~(\ref{eq: T disc},\ref{eq: C disc}) satisfy
          \begin{align} \label{eq: steering ineq disc}
 |\Cs_n|  \leq 
 \frac{\sin \left( \frac{\pi}{2} \Ts_n \right)}{n \sin \left(\frac{\pi }{2 n}\right)}
     \end{align}
\end{corollary}
\begin{proof}
    The bound simply follows from the fact that the function $\frac{\sin \left( \frac{\pi}{2} \Ts_n \right)}{n \sin \left(\frac{\pi }{2 n}\right)}$ is concave
and equals to $\widehat{\Cs}_{k|n}$ at $\Ts_n = \widehat{\Ts}_{k|n}$. In other words, the boundary of the LHS set is a piece-wise linearization of this concave function.
\end{proof}

Two things are worth noting. First, the function $\frac{\frac{\pi}{2n}}{\sin(\frac{\pi}{2n})}\geq 1$ is decreasing with $n$ and tends to $1$ in the limit $n\to \infty$, such that the bound on $\Cs_n$ becomes tighter. Second, in this limit this bound converges to the one of~\eqref{eq app: cont steering witness}, for continuous settings.

\subsection{Steering inequality with continuous input $\theta$}

Here, we discuss the case where the input of \BobL\, is continuous $y\in [0,\pi)$, which corresponds to the limit $n\to \infty$.

\label{app: steering cont}

\begin{lemma}[Steering witness for arbitrary loss]
    All LHS models ($\rho_{b|\theta} \stackrel{\rm LHS}{=} \sum_\lambda p(b|\theta,\lambda) \rho_\lambda $) satisfy the tight inequality 
    \be\label{eq app: cont steering witness}
|\Cs| \leq \frac{2}{\pi} \sin(\Ts \frac{\pi }{2}),
    \ee
    for the steering observables
    \begin{align}
    \Ts &= \int_0^\pi\frac{\dd \theta}{\pi} \ts (\theta) \qquad \,\, \text{with} \qquad \ts(\theta)=\tr (\rho_{0|\theta}+\rho_{1|\theta})
    \\
    \Cs &= \int_0^\pi \frac{\dd \theta}{\pi} \cs(\theta)\qquad \text{with} \qquad  \cs(\theta) = \tr[ (\cos(\theta) Z +\sin(\theta) X) (\rho_{0|\theta}-\rho_{1|\theta})].
    \end{align}  Furthermore, there exists an LHS model saturating the bound with $\ts(\theta)= t$ and $\cs(\theta) = \frac{2}{\pi} \sin(t \frac{\pi }{2})$ for any $t\in[0,1]$.
\end{lemma}

\begin{proof}
Consider any LHS model. It is represented by a (potentially continuous) state assemblage $\{\rho_\lambda =\mu_\lambda \varrho_\lambda\}$, with each state $\varrho_\lambda$ (with $\tr \varrho_\lambda=1$) sampled with some probability (density) $\mu_\lambda$, and a response function $p(b|\lambda,\theta)$.

To start let show that it is sufficient to consider LHS models with a state assemblage consisting of pure real states
\be\label{eq: pure real}
\varrho_\zeta = \Psi_{\zeta} = \frac{1}{2}(\id +\cos(\zeta) Z + \sin(\zeta) X)
\ee
with $\zeta \in[0,2\pi)$. This immediately follows from two observations. On the one hand, if any state $\rho_\lambda$ appearing the state assemblage is not real it can be replaced by the real state $\varrho_\lambda'=\frac{1}{2}(\varrho_\lambda + \varrho_\lambda^*)$ without affecting the values $\Ts$ and $\Cs$, since the quantities $\cs(\theta)$ and $\ts(\theta)$ are insensitive to complex conjugation of the states $\rho_{b|\theta}$. On the other, any real quantum state $\varrho_\lambda$ can be decomposed as a mixture of real pure quantum states, hence if $\varrho_\lambda$ appears in the assemblage it can be replaced by the corresponding set of pure real states, of the form  of Eq.~\eqref{eq: pure real} in the case of a qubit. Hence, without loss of generality we can consider LHS models with  the state assemblages of the form $\{\mu_\zeta \Psi_\zeta \}$, where $\mu_\zeta$ is a probability (density) on $[0,2\pi)$, and some response function $p(b|\theta,\zeta)$. Given the symmetry of the quantities $\Cs$ and $\Ts$ it is also not difficult to show that the density of states can be assumed uniform $\mu_\zeta = \frac{1}{2\pi}$, however we do not require this to proceed with the proof.

Let us now fix a specific hidden state $\Psi_\zeta$ in the assemblage, and ask what are the possible values of $\Cs^{{\zeta}}$ and $\Ts^{(\zeta)}$ achievable for this specific state. Using $\tr \Psi_\zeta (\cos(\theta) Z +\sin(\theta) X)= \cos(\theta-\zeta)$ we find
\begin{align}
\ts^{(\zeta)}(\theta) &= p(0|\theta,\zeta) +  p(1|\theta,\zeta) \\
\cs^{(\zeta)}(\theta) &=\cos(\theta-\zeta)\,  \big(p(0|\theta,\zeta) -  p(1|\theta,\zeta)\big),
\end{align}
which can be simply summarized as
\be
|\cs^{(\zeta)}(\theta) | \leq |\cos(\theta-\zeta)| \ts^{(\zeta)}(\theta) \qquad \text{and} \qquad 0\leq \ts^{(\zeta)}(\theta) \leq 1,
\ee
where all the inequalities can be saturated by the appropriate choice of $p(0|\theta,\zeta)$ and $p(1|\theta,\zeta)$. Plugging this in the average steering quantities $\Ts^{(\zeta)} =\int_0^\pi\frac{\dd \theta}{\pi} \ts^{(\zeta)} (\theta)$ and $\Cs^{(\zeta)} =\int_0^\pi\frac{\dd \theta}{\pi} \cs^{(\zeta)} (\theta)$ we obtain the saturable bounds 
\begin{align}
- \int_0^\pi\frac{\dd \theta}{\pi} |\cos(\theta-\zeta)|\, \ts^{(\zeta)} (\theta)\leq \Cs^{(\zeta)} \leq \int_0^\pi\frac{\dd \theta}{\pi} |\cos(\theta-\zeta)|\,  \ts^{(\zeta)} (\theta),
\end{align}
where $\delta = \theta-\zeta \in [-\zeta, \pi-\zeta]$, i.e. $|\cos(\theta-\zeta)|$ runs over its period of $\frac{\pi}{2}$ twice.
Here, it is direct to see that the strategy maximizing $\pm \Cs^{(\zeta)}$ for a given $\Ts^{(\zeta)}$ consist of setting $\ts^{(\zeta)}(\theta)=1$ on the domain where $|\cos(\theta-\zeta)|$ is larger than a certain value and $\ts^{(\zeta)}(\theta)=0$ where it is smaller. Since $|\cos(\delta)|$ is a decreasing function on $[0,\pi/2]$,  this strategy achieves the values
\begin{align}
 \Cs^{(\zeta)} &= \pm \int_0^\Omega \frac{ \dd \delta}{\pi/ 2} |\cos(\delta)| = \frac{2}{\pi} \sin(\Omega)\\
 \Ts^{(\zeta)} &= \int_0^\Omega \frac{ \dd \delta}{\pi/ 2} = \frac{2}{\pi} \Omega,
\end{align}
for a parameter $ 0\leq \Omega \leq \frac{\pi}{2}$. Getting rid of this parameter we obtain the following saturable bound
\be
-\Cs^{(\zeta)} \leq \frac{2}{\pi} \sin( \Ts^{(\zeta)} \frac{\pi}{2}) \leq \Cs^{(\zeta)}.
\ee

Since this bound must hold for all state state $\Psi_\zeta$, and the function  $\frac{2}{\pi} \sin( \Ts^{(\zeta)}\frac{\pi}{2})$ is concave by construction (as picking the highest values of cosine first), it follows that the bound must also hold for the average quantities $\Cs = \int \mu(\zeta) \dd \zeta \Cs^{(\zeta)}$ and $\Ts = \int \mu(\zeta) \dd \zeta \Ts^{(\zeta)}$ and remain saturable, concluding the main part of the proof.  Finally to see that the quantities $\ts(\theta)$ and $\cs(\theta)$ can be made independent of $\theta$, simply consider the uniform state assemblage $\mu(\zeta)= \frac{1}{ 2\pi}$ and think of their symmetry.
\end{proof}

\section{CHSH as self-test of Alice's measurements.} \label{app:CHSH}

Let $\rho_{AB}$ be the global state shared by Alice and Bob. Since Alice's measurements are binary, it is convenient to define the two corresponding  (Hermitian) observables as
\be
A_x = {\rm A}_{0|x} - {\rm A}_{1|x}.
\ee
A priori these measurements need not be projective, however one can consider the dilated state $\ket{\rho_{\tilde{A}  A B}}$ with  $\tr_{\tilde{A}} \rho_{\tilde{A} A B} =\rho_{AB}$, on which the measurements are projective, ${\rm A}_{a|x}={\rm A}_{a|x}^2$. We thus assume that the measurements are projective, and the operators $A_x$ have eigenvalues $\pm 1$. Similarly, without loss of generality we can assume that the global state $\rho_{AB}=\ketbra{\Psi}_{AB}$ is pure. Formally, this can be done considering the purification of any state, and absorbing the purifying system as Bob's subsystem, on which his measurements act trivially.  \\

By virtue of Jordan's lemma~\cite{Jordan1875, Scarani2012}, there exists a basis of Hilbert space associated to Alice's quantum system such that
\be\label{eq: Jordan A}
A_x = \bigoplus_{ \alpha_i} \left(\cos(\alpha_i) H +(-1)^x \sin( \alpha_i) M \right) 
\ee
where 
 $H =\frac{Z+X}{\sqrt{2}}$, $M =\frac{Z-X}{\sqrt{2}}$, and $X$, $Z$ are the Pauli operators.  We are free to choose the basis inside each qubit block such that $\alpha \in [0,\pi/2]$, so that the observables go from parallel to antiparallel Pauli operators. This guarantees that  $\cos(\alpha)$ and $\sin(\alpha)$ are both positive. \\

Now, let $\Pi_{i}$ be the projectors on the qubits blocks $i$ in the Jordan decomposition (with measurement angle $\alpha_i$). By applying the local decoherence map, Alice can prepare the state
$
\bar \rho_{AB} = \sum_{i} \Pi_{i} \rho_{AB} \Pi_{i} = \bigoplus_{i} \mu_i\,\ketbra{\psi_i}_{AB},
$
where $\mu_i = \tr \Pi_i \rho_{AB} \Pi_i$ is a probability density and $\ket{\psi_i}_{AB}=\frac{1}{\sqrt{\mu_i}}\Pi_i \ket{\Psi}_{AB}$ is a pure state with a qubit on Alice's side.
The state $\bar \rho_{AB}$ is a local post-processing of $\rho_{AB}$ (in particular if it is entangled, $\rho_{AB}$ must be), and the two are indistinguishable in any setup when Alice measurements are restricted to $A_x$. For our purpose, we can identify the two states and simply write 
\be\label{eq: Jordan rho}
\rho_{AB} =\bigoplus_{i} \mu_i\,\ketbra{\psi_i}_{AB},
\ee

Consider an SP test performed with the measurements $\{A_{a|x}\}_x$ and $\{B_{b|y}^{\rm S}\}_{y}$ by Alice and \BobS, respectively. 
The CHSH score observed in the test 
\be
\S= \sum_{a,b,x,y} (-1)^{a+b+xy}\, p(a,b|x,y,{\rm S}) 
\ee
can be written in the form $\S = \tr \bar \rho_{AB} \Big((A_0+A_1)B^{\rm S}_0+ (A_0-A_1)B^{\rm S}_1 \Big)$
where we defined the observables $B^{\rm S}_{y} = {\rm B}^{\rm S}_{0|y}-{\rm B}^{\rm S}_{1|y}$. 
With the help of the Jordan form in Eq.~(\ref{eq: Jordan A},\ref{eq: Jordan rho}) we obtain 

\begin{align}
A_0+A_1 &= 2  \bigoplus_i \cos(\alpha_i) H 
\\
A_0-A_1 &= 2 \bigoplus_i \sin( \alpha_i) M,
\end{align}
and can rewrite the expected CHSH score as
\begin{align}
\S& = \sum_i \mu_i \, \S_\alpha \qquad \text{with} \qquad
\S_\alpha =  2\bra{\psi_i} \cos(\alpha_i) H \otimes B_0^{\rm S} + \sin(\alpha_i) M \otimes B_{1}^{\rm S} \ket{\psi_i} \leq 2 (\cos(\alpha_i)+ \sin( \alpha_i)), 
\end{align}
where we used $\|H \otimes B_0^{\rm S}\|, \|M \otimes B_1^{\rm S}\|\leq 1$.  We summarize the conclusions of the last two sections in the following Lemma. 

\begin{lemma} \label{lemma: CHSH}
In the CHSH test, Alice's binary measurements and the state can be decomposed as 
\begin{align}
A_x &= \bigoplus_{i} A_x^{\alpha_i} \quad \text{with} \quad A_x^{\alpha} := \cos( \alpha) H +(-1)^x \sin(\alpha) M \\
\rho_{AB} &=\bigoplus_{i} \mu_i \ketbra{\psi_i},
\end{align}
with $ \alpha_i \in\left[0,\frac{\pi}{2}\right]$,  $H =\frac{Z+X}{\sqrt{2}}$ and $M =\frac{Z-X}{\sqrt{2}}$.
Observing a CHSH score of $\S$ guarantees that
\be\label{eq: LB angle}
\S \leq  \sum_{i} \mu_{i} \,2 \big(\cos(\alpha_i) +\sin( \alpha_i)\big).
\ee
\end{lemma}

\section{Proof of the main result}\label{app: proof main}

In this section we complete the proof of the main result Eq.~\eqref{eq:main} sketched in the main text. We start by giving a general linking between steering inequalities and routed Ball tests.

\subsection{Routed Bell tests from steering inequalities}

First, rewriting the definition of the SRQ correlations in Eq.~\eqref{eq: def SRQ} using the Jordan's decomposition of Alice's measurements gives 

\begin{align} 
p(a,b|x,y,r) &\stackrel{\rm SRQ}{=} \begin{cases}
     \sum_i \mu_i \tr \big( \A_{a|x}^{\alpha_i} \otimes (\sum_\lambda p(b|y,\lambda) \,  {\rm E}_\lambda ) \,\ketbra{\psi_i}_{AB} \big) & r = L\\
    \sum_i \mu_i \tr (\A_{a|x}^{\alpha_i} \otimes \B_{b|y}^{\rm S} \, \ketbra{\psi_i}_{AB} )&  r = S  
    \end{cases}\\
    &{\rm with} \qquad \A_{a|x}^\alpha := \frac{1}{2}(\id + (-1)^a (\cos(\alpha) H +(-1)^x \sin(\alpha) M)).
\end{align}
For each value of $\alpha$ the expression in the long pass defines a local hidden state (LHS) model~\cite{Wiseman2007,Uola2020,Skrzypczyk2015} 
\begin{equation}
    p_i(a,b|x,y,L) \stackrel{\rm LHS}{:=}  \sum_\lambda p(b|y,\lambda) \tr \big( \A_{a|x}^{\alpha_i} \, \rho_i(\lambda) \big),
\end{equation}
with Alice's hidden states given by $\rho_i(\lambda) := \tr_B (\ketbra{\psi_i}_{AB} \, {\rm E}_\lambda)$, such correlations $p_i(a,b|x,y,L)$ are called {\it unsteerable} for the trusted measurements ${\rm A}^{\alpha_i}_{a|x}$.  In turn, the CHSH score $\S$ observed in the short pass bounds~\eqref{eq: LB angle} the distribution of the angle $\alpha$. Hence,
\begin{align}\label{eq: step alsdf}
  p(a,b|x,y,r) \, \,\text{is SRQ } \implies  
  \begin{cases}
 p(a,b|x,y,L) = \sum_i \mu_i \, p_i(a,b|x,y,L), \text{ unsteerable for } {\rm A}^{\alpha_i}_{a|x}  \\
 \sum_i \mu_i {\rm s}_{\alpha_i}  \geq \S \\
  \end{cases}
\end{align}
with ${\rm s}_\alpha = 2 (\cos(\alpha)+\sin(\alpha))$.\\

 This expression  allows one to promote steering inequalities to routed Bell tests. More precisely, consider a {\it family of steering inequalities} ${\rm T}\big[p(a,b|x,y) ;{\rm s}_\alpha\big]\leq 0$ for all measurement angles $ \alpha \in\left[0,\frac{\pi}{2}\right]$, i.e. a scalar function such that 
\begin{equation}\label{eq: steering familily}
   p(a,b|x,y)   \text{ is unsteerable for } {\rm A}_{a|x}^\alpha \implies \quad {\rm T}\big[p(a,b|x,y) ;{\rm s}_\alpha\big] \leq 0.
\end{equation}
Combining the steering inequalities with Eq.~\eqref{eq: step alsdf} gives the following lemma.

\begin{lemma}
     \label{app:steering_to_bell}
In the routed scenario (Fig.~\ref{fig:1}) all SRQ correlations $p(a,b|x,y,r)$ (Eq.~\eqref{eq: def SRQ}) satisfy the routed Bell inequality
\begin{align}
\widehat{\rm T}\big[p(a,b|x,y,L) ;\S \big] \leq 0,
\end{align}
    where $\S \geq 2$ is the CHSH score in the short path, and $\widehat{\rm T}[\bm p; {\rm s}]$ is any scalar function which is
    \begin{itemize}
        \item[(i)] a lower bound for a family of steering inequalities (Eq.~\ref{eq: steering familily}), i.e. $\widehat{\rm T}[\bm p; {\rm s}]\leq {\rm T}[\bm p; {\rm s}]$ \\
        \item[(ii)]  convex, i.e. $\widehat{\rm T}[\lambda \bm p +(1-\lambda)\bm p';\lambda {\rm s} +(1-\lambda){\rm s}']\leq \lambda \widehat{\rm T}[ \bm p ; {\rm s}]+(1-\lambda) \widehat{\rm T}[ \bm p' ; {\rm s}']$\\
        \item[(iii)] increasing in $\rm s$, i.e. $\widehat{\rm T}[\bm p; {\rm s}]\leq \widehat{\rm T}[\bm p; {\rm s}']$ for ${\rm s}\leq {\rm s'}$.\\
    \end{itemize}
\end{lemma}
\begin{proof}
To shorten the notation we denote $\bm p = p(a,b|x,y,L)$ and $\bm p_i = p_i(a,b|x,y,L)$. All SRQ correlations satisfy
\begin{align}
\begin{cases} \bm p = \sum_i \mu_i \, \bm p_{i}, \text{ unsteerable for } {\rm A}^{\alpha_i}_{a|x}  \\
 \S\leq \sum_i \mu_i\, {\rm s}_{\alpha_i}  \\
  \end{cases}.
\end{align}
Since all $\bm p_i$ are unsteerable for ${\rm A}^{\alpha_i}_{a|x}$, and $\rm T$ are steering inequalities, we have ${\rm T}\big[\bm p_i ;{\rm s}_{\alpha_i} \big]\leq 0$. Using the properties $(i-iii)$  we obtain
\begin{equation}
   \widehat{\rm T}\big[\bm p;\S \big] \underset{(iii)}{\leq} \widehat{\rm T}\big[\bm p ;\sum_i \mu_i\, {\rm s}_{\alpha_i} \big]  = \widehat{\rm T}\big[\sum_i \mu_{i} \,\bm p_{i} ;\sum_i \mu_i\, {\rm s}_{\alpha_i} \big] \underset{(ii)}{\leq} \sum_i \mu_i \, \widehat{\rm T}\big[\bm p_{i} ;{\rm s}_{\alpha_i} \big] \underset{(i)}{\leq} \sum_i \mu_i {\rm T}\big[\bm p_{i} ;{\rm s}_{\alpha_i} \big]\leq 0,
\end{equation}
Note that when $\S<2$, one can use the algebraic bound $\sum_\alpha \mu_i \, {\rm s}_{\alpha_i}\geq 2$ in the first inequality. That is, the routed Bell inequity is valid with $\widehat{\rm T}\big[p(a,b|x,y,L) ;2 \big] \leq 0$.
\end{proof}

\subsection{The main result}

We now give a proof of the general result, in a slightly more general form as compared to the main text. Recall that our routed Bell inequality given in terms of the two long path quantities
\begin{align}
\C_n &= \frac{1}{n}\sum_{y=0}^{n-1}\sum_{a,b=0}^1 (-1)^{a+b} \Big( \cos{\theta_y} \,  p(a,b|0,y,L) + \sin\theta_y \, p(a,b|1,y,L)\Big),
\\
    \T_n &= \frac{1}{n}\sum_{y=0}^{n-1} \sum_{b=0,1} p(b|y,r=L) \,.
\end{align}
and the CHSH score $\S= \sum_{a,b,x,y} (-1)^{a+b+x y}\, p(a,b|x,y,S)$ in the short path. 

\begin{result}[Routed Bell inequality] \label{result app :main}
(i) For $\S\geq 2$, all SRQ correlations satisfy the routed Bell inequality
     \be \label{app: main} 
   n\sin\left(\frac{\pi}{2 n}\right)\C_n\leq  \widehat{\rm R }(\T_n,\S) ,
    \ee
    where $\widehat{\rm R }(t,s)$ is the concave envelope of the function ${\rm R }(t,s):=\sqrt{2} \sin(\frac{\pi}{2}{t}) \gamma(s)= \sin(\frac{\pi}{2}{t})\frac{s + \sqrt{8 - 
    s^2}}{2\sqrt{2}}$
     on the domain $(t,s)\in[0,1]\times [2,2 \sqrt 2]$.

(ii) Write $\epsilon=1-s/(2\sqrt{2})$ and $u=\pi t/2$.
On the full domain $(t,s)\in[0,1]\times[2,2\sqrt{2}]$, one has
${\rm R}(t,s)=\widehat{\rm R}(t,s)$ if and only if
\begin{equation}\label{app eq: cond conv}
2(1-\epsilon)\sin u
\geq \left[2+\sqrt{(2-\epsilon)\epsilon}
-\epsilon(5-2\epsilon)\right]u\cos u.
\end{equation}
For $0<t<1$, this is equivalent to
\[
\frac{\tan u}{u}\geq
\frac{2+\sqrt{(2-\epsilon)\epsilon}-\epsilon(5-2\epsilon)}{2(1-\epsilon)}.
\]
At $t=0$ and $t=1$, condition~\eqref{app eq: cond conv} is automatically satisfied.
Moreover, $\frac{\pi}{2}t\geq{\rm R}(t,s)$ is sufficient for equality with the
concave envelope. For $t>0$, this sufficient condition is equivalent to
$\frac{\pi t/2}{\sin(\pi t/2)}\geq\sqrt{2}\,\gamma(s)$.

(iii) All SRQ correlations satisfy the following family of linear routed Bell inequalities
\begin{equation}\label{eq app: lin Bell}
   \forall \beta \in (0,1]\,, \xi \in [1/\sqrt{2},1]\,, \quad \frac{n\sin(\frac{\pi}{2n})}{\beta} \WWn -  \frac{\pi}{2} \TTn  \le h_{\beta}(\xi) +\left(\frac{\chsh}{2\sqrt{2}} - \xi \right)\, h'_\beta(\xi) \,,
\end{equation}
where,
\begin{equation}
    h_{\beta}(x) := \sqrt{ f_{\beta}(x)^2-1}\,\, - \arccos({1}/{f_{\beta}(x)})\,, \quad \quad f_{\beta}(x) := \frac{x+\sqrt{1-x^2}}{\beta}.
\end{equation}

(iv) The bounds~(\ref{app: main},\ref{eq app: lin Bell})  are tight in the limit of continuous input $n\to \infty $, i.e. there is an explicit SRQ model saturating them.
\end{result}

The statement of the result in the main text is obtained by combining $(i)$ with $(ii)$: SRQ correlations satisfy  $n\sin\left(\frac{\pi}{2 n}\right)\C_n\leq  \rm R (\T_n,\S)$ for $\T_n>0$ and $\frac{\pi}{2}\T_n/\sin\left(\frac{\pi}{2}\T_n\right) \geq   \sqrt{ 2} \, \gamma(\S)$.
For $\T_n=0$, both sides vanish.
 In the next sections, we prove the assertions in the following order {\it (i),(iii),(ii),(iv)}.

\subsection{(i) A general  bound on the SRQ correlations}
\label{app:Result1}

Consider the family of steering inequalities
\be{\rm T}[p(a,b|x,y);{\rm s}_\alpha] = n\sin\left(\frac{\pi}{2 n}\right)\C_n^\alpha - {\rm R}(\T_n,{\rm s}_\alpha) =n\sin\left(\frac{\pi}{2 n}\right)\C_n^\alpha-\sqrt{2} \sin\left(\frac{\pi}{2}\T_n\right) \gamma\left(\sb_\alpha \right)\leq 0 \label{appeq:AJordan}
\ee
derived in App.~\ref{app: stearing from main}. Note that ${\rm T}[p(a,b|x,y);{\rm s}]$ is monotonically increasing in ${\rm s}\in[2,2\sqrt 2]$ since  $\sin\left(\frac{\pi}{2}\T_n\right)$ is positive and $ \gamma\left(\sb_\alpha \right)$ is monotonically decreasing.

Let $\widehat{\rm T}[p(a,b|x,y);{\rm s}_\alpha]$ be the {\it convex envelope} of this function~\cite{Boyd2004}. By definition, it is the maximal convex function such that 
\be
\widehat{\rm T}[p(a,b|x,y);{\rm s}]\leq {\rm T}[p(a,b|x,y);{\rm s}].
\ee
In addition it is also monotonically increasing in $\sb$~\cite{Boyd2004}. Therefore, $\widehat{\rm T}[p(a,b|x,y);{\rm s}]$ satisfies the assumptions $(i-iii)$ of the lemma~\ref{app:steering_to_bell}, which guarantee that  
\be
\widehat{\rm T}[p(a,b|x,y);\S]\leq 0
\ee
is a routed Bell inequality. Finally note that 
\be
\widehat{\rm T}[p(a,b|x,y);\S] = n\sin\left(\frac{\pi}{2 n}\right)\C_n - \widehat{\rm R}(\T_n,\S),
\ee
where $\widehat{\rm R}(t,{\rm s})$ is the concave envelope of ${\rm R}(t,{\rm s})$ (equivalently $-\widehat{\rm R}(t,{\rm s})$ is the convex envelope of ${\rm R}(t,{\rm s})$). Hence
\be
n\sin\left(\frac{\pi}{2 n}\right)\C_n \leq  \widehat{\rm R}(\T_n,{\rm s}_\alpha)
\ee
is a routed Bell inequality, proving the assertion $(i)$.

\subsection{(iii) A family of routed Bell tests linear in the LP quantities}

\label{app: Linear RBT}

In the following, we will derive a family of linear routed Bell inequalities for SRQ correlations. Let ${\rm E}_{n,\beta} (\S) $ denote the maximum value of the expression below that can be achieved by SRQ correlations when a CHSH value of at least $\chsh$ is observed in the SP, i.e.,
\begin{equation}\label{appeq:routed_bell_inequality1}
    \WWn - {\lambda_n}\,{\beta}\, \frac{\pi}{2} \,\TTn \le {\rm E}_{n,\beta} (\S) , \,
\end{equation}
where $\beta \in (0,1]$ is a constant, $\lambda_n := {1}/({n\sin(\frac{\pi}{2n})})$. Using the Jordan decomposition, and the steering inequality (Eq.~\eqref{appeq:AJordan}) $\WWn^{\alpha} \le \lambda_n  \sqrt{2} \sin\left(\frac{\pi}{2}\TTn \right) \gamma(\sb_{\alpha})$\,, we get
\begin{align}
   \WWn - {\lambda_n}\,{\beta}\, \frac{\pi}{2}\, \TTn &=
   \sum_i  \mu_i \left[\WWn^{\alpha_i} -{\lambda_n}\,{\beta}\,\frac{\pi}{2}\,\TTn^i \right] \\
   &\le \sum_i  \mu_i \left[ \lambda_n  \sqrt{2} \sin\left(\frac{\pi}{2}\TTn^i\right) \gamma(\sb_{\alpha_i}) -{\lambda_n}\,{\beta}\, \frac{\pi}{2} \,\TTn^i \right] \, \label{appeq:routed_bell_inequality2}.
\end{align}
To simplify the notation, let us define 
\begin{align}
    C_{\alpha_i} := \frac{\sb_{\alpha_i}}{2\sqrt{2}}\,, \quad f_{\beta}(x) := \frac{\sqrt{2}}{\beta} \gamma(2 \sqrt{2} \,x ) = \frac{x+\sqrt{1-x^2}}{\beta}, \quad \tilde{ \TTn}^i := \frac{\pi}{2}\TTn^i\,,
\end{align}
which gives us the following bound on $C_{\alpha_i}$, from~Eq.~\eqref{eq: LB angle}, due to the observed value of $\chsh$ in the SP
\begin{equation}\label{appeq:SPtest}
    \sum_i \mu_i C_{\alpha_i} = \frac{1}{2\sqrt{2}}\sum_i \mu_i \sb_{\alpha_i} \ge \frac{\mathcal{S}}{2\sqrt{2}}\,.
\end{equation}
This lets us rewrite Eq.~\eqref{appeq:routed_bell_inequality2} as
\begin{equation}\label{eq:SRQ_bound1}
    \WWn - {\lambda_n}\,{\beta}\, \frac{\pi}{2}\, \TTn \le \lambda_n \, \beta \,\sum_i \mu_i \left[ f_{\beta}(C_{\alpha_i}) \sin(\tilde \TTn^i) - \tilde \TTn^i  \right]\,.
\end{equation}
Note that since $C_{\alpha_i} \in [1/\sqrt{2},1]$ and $\beta \in (0,1]$, we have $f_{\beta}(C_{\alpha_i}) \in [1, \infty)$. Furthermore, the family of functions 
\begin{equation}
    g_{\alpha,\beta}(\tilde \TTn):=f_{\beta}(C_{\alpha}) \sin(\tilde \TTn) -  \tilde \TTn \,,
\end{equation}
defined on the domain $\tilde \TTn \in [0,\pi/2]$, are differentiable and admit a unique maximum at the point $\tilde \TTn =  \arccos\left({1}/{f_{\beta}(C_{\alpha})}\right) $. Substituting this value of $\tilde \TTn $ in Eq.~\eqref{eq:SRQ_bound1}, and using the identity $\sin(\arccos(x)) = \sqrt{1-x^2}$, we obtain
\begin{equation}
   \WWn - {\lambda_n}\,{\beta}\, \frac{\pi}{2}\, \TTn  \le \lambda_n \, \beta \, \sum_i \mu_i \left[ \sqrt{f_{\beta}(C_{\alpha_i})^2-1}\,\, - \arccos({1}/{ f_{\beta}(C_{\alpha_i})}) \right]\,.
\end{equation}
To proceed further, we note that the family of functions 
\begin{equation}\label{appeq:h}
    h_{\beta}(x) := \sqrt{ f_{\beta}(x)^2-1}\,\, - \arccos({1}/{f_{\beta}(x)})\,
\end{equation}
are nonincreasing and concave for $x \in [1/\sqrt{2},1]$, i.e., $h_{\beta}'(x) \le 0$ and $h_{\beta}''(x) \le 0$ whenever $\beta \in (0,1]$ (see App.~\ref{app:concavity}). Therefore,
\begin{align}
        \WWn - \lambda_n \, \beta \, \frac{\pi}{2}\, \TTn  &\le \lambda_n\, \beta \, \sum_i \mu_i\, h_{\beta}(C_{\alpha_i}) \\
        &\le \lambda_n \, \beta \, h_{\beta}\left(\sum_i \mu_i C_{\alpha_i}\right)\,\\
        &\le \lambda_n \, \beta \, h_{\beta}\left(\frac{\mathcal{S}}{2\sqrt{2}}\right) := {\rm E}_{n,\beta} (\S) \label{appeq:routed_bell_inequality3},
\end{align}
where we have used the fact that $h_{\beta}(C_{\alpha})$ is a concave nonincreasing function, and is hence maximized for the minimum value of its argument, which is given by Eq.~\eqref{appeq:SPtest}.

To obtain a linear routed Bell inequality from above, we only need to note that since $h_{\beta}(x)$ is a concave function, it is upper bounded by the family of tangents at every point of it argument. Explicitly, 
\begin{equation}
    \forall \xi \in [1/\sqrt{2},1]\,, \quad  h_{\beta}\left(\frac{\mathcal{S}}{2\sqrt{2}}\right) \le h_{\beta}(\xi) + h'_\beta(\xi) \left(\frac{\chsh}{2\sqrt{2}} - \xi\right)\,,
\end{equation}
which, when combined with Eq.~\eqref{appeq:routed_bell_inequality3}, leads to the linear routed Bell inequality
\begin{equation}
   \forall \beta \in (0,1]\,, \xi \in [1/\sqrt{2},1]\,, \quad \frac{n\sin(\frac{\pi}{2n})}{\beta} \WWn -  \frac{\pi}{2} \TTn  \le h_{\beta}(\xi) + \left(\frac{\chsh}{2\sqrt{2}} - \xi \right)\,h'_\beta(\xi) \,.
\end{equation}
We can use these linear inequalities to arrive at the main result, as we discuss below.

\subsubsection{Using the linear inequalities to characterize the SRQ set}
\label{sec: lin to conv}
To arrive at the main result in the form of Eq.~\eqref{eq:main}, we first note that Eq.~\eqref{appeq:routed_bell_inequality3} can be written as
\begin{equation}\label{eq:routed_bell_inequality2}
    \forall \beta \in (0,1], \quad \frac{n\sin(\frac{\pi}{2n})}{\beta} \WWn - \frac{\pi}{2} \TTn \le  \sqrt{\left({\sqrt{2}\, \gamma({\chsh})}/{\beta}\right)^2-1}\,-\arccos\left({\beta}/{(\sqrt{2} \,\gamma(\chsh)})\right).
\end{equation}
We can now use this family of inequalities to bound the set of allowed values of $\WWn$, for a given value of $\TTn$ and $\chsh$. To this end, let us first reparametrize $\beta$ in terms of $u$, defined by
\begin{equation}
    \beta = \sqrt{2} \,\gamma(\chsh) \cos(u)\,, u \in [0,\frac{\pi}{2}),
\end{equation}
which is valid whenever $\sqrt{2} \,\gamma(\chsh) \le 1/\cos(u)$. Substituting this into Eq.~\eqref{eq:routed_bell_inequality2}, we obtain
\begin{equation}\label{eq:routed_bell_inequality3}
    {n\sin(\frac{\pi}{2n})} \WWn \le \sqrt{2}\, \gamma(\chsh) \left[ \frac{\pi}{2} \TTn \cos(u) + \sin(u) - u \cos(u) \right].
\end{equation}
Although the upper bound on $\WWn$ above is valid for every choice of $u \in [0,\frac{\pi}{2})$ and $\sqrt{2} \,\gamma(\chsh) \le 1/\cos(u)$, it is convenient to eliminate the parameter $u$ by finding the lowest upper bound. This is achieved by minimizing the RHS of the above equation over $u$, which gives us $u = (\pi/2) \TTn$. We substitute this in the equation above to obtain
\begin{equation}\label{eq:routed_bell_inequality4}
    {n\sin(\frac{\pi}{2n})} \WWn \le \sqrt{2}\, \gamma(\chsh) \, \sin(\frac{\pi}{2}\TTn) \,,
\end{equation}
valid whenever 
\begin{equation} \label{eq: conc lin}
    (\ast\ast\ast\ast) \quad \quad \sqrt{2} \, \gamma(\chsh) \le 1/\cos(\frac{\pi}{2}\TTn)\,.
\end{equation} 
This region is denoted by $(\ast\ast\ast\ast)$ in Fig.~\ref{fig:conc}, and is strictly greater than the region of interest in the statement of the result in Eq.~\eqref{eq:main}. This implies ${\rm R }(\T_n,\cS)=\widehat{\rm R }(\T_n,\S)$ on the range of values given by Eq.~\eqref{eq: conc lin}. However, we will now see that the function $\rm R$ coincides with its concave envelope on a larger range of values.

\subsection{(ii) Properties of the concave envelope  $\widehat{\rm R}(t,s)$ }

\label{app: convex F}

In this section we discuss the properties the function ${\rm R }(t,s)=\sin(\frac{\pi}{2}{t})\frac{s + \sqrt{8 - 
    s^2}}{2\sqrt{2}}$  and its concave envelope $\widehat{\rm R }(t,s)$ on the domain $(t,s) \in [0,1]\times[2,2\sqrt 2]$. Precisely, we derive a necessary and sufficient condition for the equality ${\rm R }(t,s)=\widehat{\rm R }(t,s).$

It will be convenient to perform a change of variables $s = 2\sqrt 2(1 -\epsilon)$, such that $\epsilon\in\left[0,1-\frac{1}{\sqrt 2}\right]$ is small when the CHSH score is close to maximal value, and discuss the properties of the concave envelope $\widehat{ F}(\epsilon,t)$ of the function  
\begin{align}
F(\epsilon,t) &:= \frac{2}{\pi} {\rm R}\left(t, 2\sqrt 2(1 -\epsilon)\right) = f(\epsilon)g(t) \qquad \text{with}\\
f(\epsilon) &:= 
1 +\sqrt{2\epsilon-\epsilon^2} - \epsilon 
\\
g(t)&:=\frac{2}{\pi} 
    \sin(t \frac{\pi}{2})
\end{align}
on the domain  $(\epsilon,t)\in [0,\epsilon_{max}]\times [0,1]$ with $\epsilon_{max}:=1-\frac{1}{\sqrt 2}$. Linear change of variables and multiplication by a constant do not affect the concavity of a function, hence 
\be
\widehat{F}(\epsilon, t) = \frac{2}{\pi} \, \widehat{\rm R}\left(t, 2\sqrt 2(1 -\epsilon)\right).
\ee  

We now show that the function $F(\epsilon,t)$ coincides with its concave envelope  $\widehat{F}(\epsilon,t)$ on a region. The final result is graphically illustrated in Fig.~\ref{fig:conc}.
\begin{figure}
    \centering
    \includegraphics[width=0.5\linewidth]{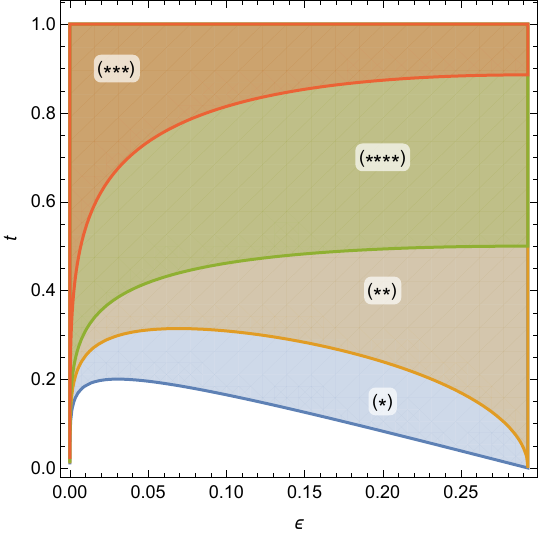}
    \caption{The shaded regions (above the corresponding lines) represent different conditions for the function $F(\epsilon,t) =\widehat{F}(\epsilon,t)$ to be equal to its concave envelope (recall that $t=\T$ and $\epsilon = 1-\frac{\S}{2\sqrt 2}$). The condition $(\ast)$ in Eq.~\eqref{eq: cond hess}, obtained by requiring that the Hessian of the function is negative semi-definite, is necessary but not sufficient at interior points. The condition  $(\ast\ast)$ in Eq.~\eqref{eq: iff}, obtained by constructing the tangent hyperplanes to the function, is shown to be necessary and sufficient at interior points in the sec.~\ref{app: convex F} devoted to this. The sufficient condition $(\ast\ast\ast\ast) \cos(\frac{\pi}{2}t ) \leq  (1-\epsilon +\sqrt{(2-\epsilon)\epsilon})^{-1}$ implied by Eq.~\eqref{eq: conc lin} is an immediate consequence of the routed Bell inequalities family~\eqref{eq app: lin Bell}, as shown in sec.~\ref{sec: lin to conv}. Finally, the simple condition $(\ast\ast\ast)$ in Eq.~\eqref{eq: cond main}, which is presented in the statement of the result in the main text and verified by our qubits example, is also formally shown to be sufficient in sec.~\ref{sec: simple to conv}.
Along the boundary $t=0$, one additionally has
$F(\epsilon,0)=\widehat F(\epsilon,0)=0$ for every allowed $\epsilon$.}
    \label{fig:conc}
\end{figure}

We first consider the boundary of the domain. A point on an edge can only be
written as a convex combination of points on that same edge. Since $F$ is
concave along each edge, it is equal to its concave envelope on the boundary.
In particular, $F(\epsilon,0)=\widehat F(\epsilon,0)=0$.
At $t=0$, this also follows directly from
$0\leq F(\epsilon,t)\leq\sqrt{2}\,t$, a linear upper bound that is zero at $t=0$.
In the differential argument below, the tangent point is assumed to satisfy
$0<\epsilon_0<\epsilon_{max}$ and $0<t_0<1$.

\subsubsection{Local concavity}

To start, we compute the Hessian of the function

\begin{equation}
{\bf H}_F(\epsilon,t) =   \left(
\begin{array}{cc}
 \partial_\epsilon^2  F(\epsilon,t) & \partial_\epsilon  \partial_t F(\epsilon,t) \\
\partial_t \partial_\epsilon F(\epsilon,t) & \partial_t^2  F(\epsilon,t) \\
\end{array}\right)
=  \left(
\begin{array}{cc}
 f''(\epsilon) g(t) & f'(\epsilon) g'(t) \\
f'(\epsilon) g'(t) & f(\epsilon) g''(t) \\
\end{array}
\right),
\end{equation}
with 
\begin{align}
f'(\epsilon)&=\frac{1-\epsilon}{\sqrt{(2-\epsilon)\epsilon}}-1,
& f''(\epsilon)&=-\frac{1}{[(2-\epsilon)\epsilon]^{3/2}},\label{eq:corrected-f-derivatives}\\
g'(t)&=\cos\left(\frac{\pi t}{2}\right),
& g''(t)&=-\frac{\pi}{2}\sin\left(\frac{\pi t}{2}\right).
\end{align}
At an interior point of the domain, the Hessian has negative trace.
It is therefore negative semidefinite if and only if its determinant is nonnegative:
\begin{equation}
(\ast)\qquad {\bf H}_F(\epsilon,t)\preceq0
\quad\Longleftrightarrow\quad \det\,{\bf H}_F(\epsilon,t)\geq0.
\end{equation}
Hence, $(\ast)$ is a necessary condition for $F=\widehat F$ in the interior
of the domain.

\subsubsection{Tangent hyperplane}

Next, we consider the hyperplanes $H_{\epsilon_0,t_0}(\epsilon,t)$ tangent to the function $F(\epsilon,t)=f(\epsilon)g(t)$ at $(\epsilon_0,t_0)$. The hyperplane is given by the equation
\begin{align}
{\rm H}_{\epsilon_0,t_0}(\epsilon,t) &:= F(\epsilon_0,t_0) +  \partial_\epsilon F(\epsilon_0,t_0) (\epsilon-\epsilon_0)+  \partial_t F(\epsilon_0,t_0) (t-t_0)\\
&=f(\epsilon_0) g(t_0) +f'(\epsilon_0) g(t_0)  (\epsilon-\epsilon_0) +  f(\epsilon_0) g'(t_0) (t-t_0)
\end{align}
If this tangent hyperplane upper-bounds $F$ on the whole domain, then
$\widehat F(\epsilon_0,t_0)=F(\epsilon_0,t_0)$.
Conversely, at a point in the interior where $F=\widehat F$, a supporting hyperplane of the concave
envelope must have the same gradient as the differentiable function $F$.
It is therefore the tangent hyperplane to $F$, and upper-bounds $F$ throughout the domain.
Thus we must determine when
\begin{equation}
\Delta_{\epsilon_0,t_0}(\epsilon,t)
:=H_{\epsilon_0,t_0}(\epsilon,t)-F(\epsilon,t)
\end{equation}
is nonnegative on $[0,\epsilon_{max}]\times[0,1]$.
Its first derivatives are

\begin{align}
    \partial_\epsilon \Delta_{\epsilon_0,t_0}(\epsilon,t) &= f'(\epsilon_0)g(t_0) - f'(\epsilon)g(t)\\
     \partial_t \Delta_{\epsilon_0,t_0}(\epsilon,t) &= f(\epsilon_0)g'(t_0) - f(\epsilon)g'(t).
\end{align}
We also remark that the Hessian of the function is given by $ {\bf H}_\Delta (\epsilon,t) =  -{\bf H}_F (\epsilon,t) $ and is independent of $(\epsilon_0,t_0)$.

\subsubsection{The values of $\Delta_{\epsilon_0,t_0}$ on the boundary of the domain.}

We first look for the global minimum of $\Delta_{\epsilon_0,t_0}$ on the boundary of the domain. Noting that $g(0)=0$, $g'(1)=0$, $f'(\epsilon_{max})=0$ and $\lim_{\epsilon\to0^+}f'(\epsilon)=+\infty$, we conclude the  following properties
\begin{align}
    &\partial_\epsilon \Delta_{\epsilon_0,t_0}(0,t) \to -\infty \qquad \text{if} \qquad t>0\\
    &\partial_\epsilon \Delta_{\epsilon_0,t_0}(\epsilon_{max},t) = f'(\epsilon_0)g(t_0) \geq 0 \\
    &\partial_t \Delta_{\epsilon_0,t_0}(\epsilon,1) = f(\epsilon_0)g'(t_0) \geq 0. 
\end{align}
These inequalities imply that the function can not have its global minimum at $\epsilon =\epsilon_{max}$, $t=1$ and $\epsilon=0$ if $t>0$, because it decreases in the direction of the interior. In turn, for $t=0$ we find that
 \begin{align}
 \partial_\epsilon \Delta_{\epsilon_0,t_0}(\epsilon,0) = f'(\epsilon_0)g(t_0) \geq 0,
 \end{align}
 i.e.  $\Delta_{\epsilon_0,t_0}(\epsilon,0)$ is non-decreasing with $\epsilon$ and attains its minimum on this interval at $\epsilon =0$.
 
 Hence the on the boundary of the domain the function $\Delta_{\epsilon_0,t_0}$ can only admit a global minimum at $(\epsilon,t)=(0,0)$. For our task we must thus verify that the function is nonnegative here. This adds the following constraint
 \begin{align}
     (\ast\ast) \qquad \Delta_{\epsilon_0,t_0}(0,0)\geq 0.
 \end{align}

\subsubsection{Local minima of $\Delta_{\epsilon_0,t_0}$}

Now that we settled the question of the boundary it remains to verify that the interior local minima of the function are nonnegative. Its interior stationary points are given by the equations $\partial_\epsilon \Delta_{\epsilon_0,t_0}(\epsilon,t) = \partial_t\Delta_{\epsilon_0,t_0}(\epsilon,t)=0$, whose explicit form is
\begin{align}\label{eq: diff e =0}
    \sin(\frac{\pi}{2}t) f'(\epsilon) &=  \sin(\frac{\pi}{2}t_0) f'(\epsilon_0) \quad\Longleftrightarrow\quad  \sin(\frac{\pi}{2}t)  =  \sin(\frac{\pi}{2}t_0) \frac{f'(\epsilon_0)}{f'(\epsilon)} \\
    \label{eq: diff t =0}
    \cos(\frac{\pi}{2}t) f(\epsilon) &=  \cos(\frac{\pi}{2}t_0) f(\epsilon_0) \quad \Longleftrightarrow \quad \cos(\frac{\pi}{2}t) =  \cos(\frac{\pi}{2}t_0) \frac{f(\epsilon_0)}{f(\epsilon) }
\end{align}
where one notes that the functions of the rhs are always nonnegative on the domain of interest, as  $f(\epsilon),f'(\epsilon)$, sine and cosine are nonnegative. 
By construction, $(\epsilon_0,t_0)$ is a stationary point with
$\Delta_{\epsilon_0,t_0}(\epsilon_0,t_0)=0$.
We show below that it is a strict local minimum when $(\ast\ast)$ holds.

Let us now show that the function can at most admit two local extrema. To see this we eliminate the variable $t$ by using $\cos^2(\frac{\pi}{2}t)+ \sin^2(\frac{\pi}{2}t) =1$. This leaves us with the following equation
\begin{align}\label{eq: cons e comb}
 \frac{ \sin^2\left(\frac{\pi}{2}t_0\right)(f'(\epsilon_0))^2}{\big(f'(\epsilon)\big)^2} +  \frac{\cos^2\left(\frac{\pi}{2}t_0\right) \big(f(\epsilon_0)\big)^2}{\big( f(\epsilon) \big)^2 } = 1
\end{align}
Let $L(\epsilon)$ denote the left-hand side of Eq.~\eqref{eq: cons e comb}.
It is strictly convex on $0<\epsilon<\epsilon_{max}$, as follows from
\begin{equation}\label{eq:corrected-reciprocal-fprime}
\begin{aligned}
\frac{\mathrm d^2}{\mathrm d\epsilon^2}\frac{1}{[f'(\epsilon)]^2}
&=\frac{6\left(1-\epsilon+\sqrt{(2-\epsilon)\epsilon}\right)}
{\sqrt{(2-\epsilon)\epsilon}\left(1-\epsilon-\sqrt{(2-\epsilon)\epsilon}\right)^4}\\
&=\frac{6f(\epsilon)}
{\sqrt{(2-\epsilon)\epsilon}\left(\epsilon+\sqrt{(2-\epsilon)\epsilon}-1\right)^4}>0,
\end{aligned}
\end{equation}
and
\begin{equation}\label{eq:corrected-reciprocal-f}
\begin{aligned}
\frac{\mathrm d^2}{\mathrm d\epsilon^2}\frac{1}{[f(\epsilon)]^2}
&=\frac{6[f'(\epsilon)]^2-2f(\epsilon)f''(\epsilon)}{[f(\epsilon)]^4}\\
&=\frac{8\sqrt{(2-\epsilon)\epsilon}-(12\epsilon^3-36\epsilon^2+26\epsilon-2)}
{[(2-\epsilon)\epsilon]^{3/2}\left(1-\epsilon+\sqrt{(2-\epsilon)\epsilon}\right)^4}>0.
\end{aligned}
\end{equation}
The first expression is positive because $f(\epsilon)>0$ and its denominator
is positive in the open interval. The second is positive because
$f(\epsilon)>0$ and $f''(\epsilon)<0$, so that
$6[f'(\epsilon)]^2-2f(\epsilon)f''(\epsilon)>0$.
For an interior tangent point, the coefficients in $L$ are positive.
Consequently, $L(\epsilon)=1$ has at most two solutions, and
$\Delta_{\epsilon_0,t_0}$ has at most two interior stationary points.

Let us now come back to the Eqs.~(\ref{eq: diff e =0},\ref{eq: diff t =0}). By inverting the trigonometric functions we find the curves on which the derivatives of $\Delta_{\epsilon_0,t_0}$ vanish
\begin{align}
  \partial_\epsilon \Delta_{\epsilon_0,t_0}(\epsilon,t)=0 \quad \Longleftrightarrow\quad
   t  &=  X_{\epsilon_0,t_0}(\epsilon):=\frac{2}{\pi}\arcsin\left[\frac{\sin(\frac{\pi}{2}t_0)  f'(\epsilon_0)}{f'(\epsilon)}\right] \\
 \partial_t\Delta_{\epsilon_0,t_0}(\epsilon,t) =0 \quad \Longleftrightarrow\quad
    t &= Y_{\epsilon_0,t_0}(\epsilon) := \frac{2}{\pi} \arccos \left[  \frac{ \cos(\frac{\pi}{2}t_0) f(\epsilon_0)}{f(\epsilon) } \right].
\end{align}
  Since the arguments of arcsin and arccos are positive, both functions $0\leq X_{\epsilon_0,t_0}, Y_{\epsilon_0,t_0}\leq 1$ produce values in the valid range if the arguments do not exceed one, otherwise the corresponding derivative does not vanish for the chosen value of  $\epsilon$. The interior stationary points occur at intersections of these curves with $0<\epsilon<\epsilon_{max}$ and $0<t<1$. It is easy to see that both functions are monotonically increasing in $\epsilon$, since we know that $\arcsin(x)$ with $1/f'(\epsilon)$ are both monotonically increasing, and  $\arccos(x)$ with $1/f(\epsilon)$ are both monotonically decreasing. Interior stationary points are thus given by the equation $ X_{\epsilon_0,t_0}(\epsilon) = Y_{\epsilon_0,t_0}(\epsilon)$, which can have at most two solution as we have shown.

Consider an interior tangent point satisfying $(\ast\ast)$.
Using the inequalities derived in the next subsection, we find
\begin{equation}\label{eq:strict-hessian-at-contact}
1+\tan^2u_0
>\left(\frac{\tan u_0}{u_0}\right)^2
\geq B_{\ast\ast}(\epsilon_0)\geq B_\ast(\epsilon_0),
\qquad u_0:=\frac{\pi t_0}{2}.
\end{equation}
The strict inequality follows from $\sin u_0<u_0$.
Thus $\det\,{\bf H}_F(\epsilon_0,t_0)>0$ and
${\bf H}_\Delta(\epsilon_0,t_0)$ is positive definite.
Hence $(\epsilon_0,t_0)$ is a strict local minimum.

To show that there is no other local minimum, let
$a_0:=f'(\epsilon_0)\sin u_0>0$ and
$b_0:=f(\epsilon_0)\cos u_0>0$.
Along the curve $t=Y_{\epsilon_0,t_0}(\epsilon)$, inside the domain,
where $f(\epsilon)>b_0$ and $\partial_t\Delta=0$, one has
\begin{equation}
\frac{\mathrm d}{\mathrm d\epsilon}
\Delta_{\epsilon_0,t_0}\bigl(\epsilon,Y_{\epsilon_0,t_0}(\epsilon)\bigr)
=\frac{2}{\pi}f'(\epsilon)
\left[\frac{a_0}{f'(\epsilon)}
-\sqrt{1-\frac{b_0^2}{f(\epsilon)^2}}\right].
\end{equation}
Since both terms inside the brackets are nonnegative,
\begin{equation}
\operatorname{sgn}\frac{\mathrm d}{\mathrm d\epsilon}
\Delta_{\epsilon_0,t_0}\bigl(\epsilon,Y_{\epsilon_0,t_0}(\epsilon)\bigr)
=\operatorname{sgn}\bigl(L(\epsilon)-1\bigr).
\end{equation}
Positive definiteness at the tangent point implies that the second derivative
along this curve is positive there, and hence $L'(\epsilon_0)>0$.
By strict convexity, any other solution of $L(\epsilon)=1$ lies to the left
of $\epsilon_0$ and crosses from $L>1$ to $L<1$.
It is therefore a local maximum along this curve and cannot be a local minimum
of $\Delta_{\epsilon_0,t_0}$.
We conclude that $(\epsilon_0,t_0)$ is the only local minimum in the interior of the domain.

 \subsubsection{The condition $(\ast \ast)$ implies $(\ast)$}

At a point in the interior where $F=\widehat F$, both $(\ast)$ and $(\ast\ast)$ are necessary.
Conversely, $(\ast\ast)$ implies the strict version of $(\ast)$ by the
comparison below. The preceding argument then leaves only the tangent point
and the origin as possible global minima of $\Delta_{\epsilon_0,t_0}$.
Their values are, respectively, zero and a nonnegative number.
Thus $(\ast\ast)$ is sufficient for the tangent plane to upper-bound $F$
throughout the domain, and hence
\begin{equation}
F(\epsilon_0,t_0)=\widehat F(\epsilon_0,t_0)
\quad\Longleftrightarrow\quad
\Delta_{\epsilon_0,t_0}(0,0)\geq0,
\qquad 0<\epsilon_0<\epsilon_{max},\quad 0<t_0<1.
\end{equation}

We now would like to put these constraints in a more usable form. We have
\begin{align}
 (\ast)\quad     \det \, {\bf H}_F(\epsilon,t) \geq 0 \quad  &\Longleftrightarrow \quad  f''(\epsilon)f(\epsilon)  g(t)  g''(t) - (f'(\epsilon) g'(t))^2  \geq 0 \\
(\ast \ast) \quad     \Delta_{\epsilon,t}(0,0)\geq 0   \quad  &\Longleftrightarrow \quad  
     f(\epsilon) g(t) - f'(\epsilon) g(t) \epsilon -  f(\epsilon) g'(t) t \geq 0.
\end{align}
Direct calculation allows one to write these bounds as 
\begin{align}\label{eq: cond hess}
(\ast) \quad \cos (\pi  t) &\leq \frac{1+ \epsilon  (3-2 (3-\epsilon ) \epsilon )}{1+2 \sqrt{(2-\epsilon ) \epsilon }-\epsilon  (5-2 (3-\epsilon ) \epsilon )}\\
\label{eq: iff}
(\ast \ast)\quad \frac{\tan \left(t\frac{\pi }{2}\right)}{t \frac{\pi }{2}} &\geq \frac{2 +\sqrt{(2-\epsilon ) \epsilon } -\epsilon  (5-2 \epsilon ) }{2 (1-\epsilon )}
\end{align}
Using $\cos(2x)=\frac{1-\tan^2(x)}{1+\tan^2(x)}$ we can re-express the first constraint as
\begin{align}
 (\ast)\quad  
 1+ \tan^2 \left(t\frac{\pi }{2}\right) &\geq B_\ast(\epsilon) :=1+\frac{\sqrt{(2-\epsilon ) \epsilon }-2 (2-\epsilon ) (1-\epsilon ) \epsilon }{1 +\sqrt{(2-\epsilon ) \epsilon }-\epsilon} \\
 (\ast \ast)\quad \left(\frac{\tan \left(t\frac{\pi }{2}\right)}{t \frac{\pi }{2}} \right)^2 &\geq B_{\ast\ast}(\epsilon) :=\left(\frac{2 +\sqrt{(2-\epsilon ) \epsilon } -\epsilon  (5-2 \epsilon ) }{2 (1-\epsilon )}\right)^2
\end{align}
Comparing the left hand sides of these equation we find that
\begin{align}
 1+ \tan^2 \left(t\frac{\pi }{2}\right) &\geq \left(\frac{\tan \left(t\frac{\pi }{2}\right)}{t \frac{\pi }{2}} \right)^2 \Longleftrightarrow \left(t \frac{\pi }{2}\right)^2 \geq \frac{\tan^2 \left(t\frac{\pi }{2}\right)}{1+ \tan^2 \left(t\frac{\pi }{2}\right)} = \sin^2 \left(t\frac{\pi }{2}\right).
\end{align}
In turn, let us also show that 
\begin{align}
 B_{\ast \ast}(\epsilon)\geq B_{\ast }(\epsilon)
\end{align}
Introduce the variable change $1-\epsilon = \cos(\alpha)$ and $\sqrt{(2-\epsilon ) \epsilon } = \sin(\alpha)$ with $\alpha \in [0,\frac{\pi}{4}]$, and rewrite
\begin{align}
 B_{\ast \ast}(\epsilon)- B_{\ast }(\epsilon) = \frac{\sin ^2\left(\frac{\alpha }{2}\right) (-5 \sin (\alpha )+\sin (3 \alpha )+5 \cos (\alpha )+8 \cos (2 \alpha )+\cos (3 \alpha ))}{2 \sin (\alpha )+\sin (2 \alpha )+2 \cos (\alpha )-\cos (2 \alpha )+1}
\end{align}

The denominator is strictly positive: $1-\cos(2\alpha)\geq0$,
the remaining terms are nonnegative, and $2\cos\alpha>0$.
The numerator is nonnegative because
$5(\cos\alpha-\sin\alpha)\geq0$, $8\cos(2\alpha)\geq0$, and
$\sin(3\alpha)+\cos(3\alpha)\geq0$ for $\alpha\in[0,\pi/4]$.
 It follows that the constraint $(\ast \ast)$ guarantees $(\ast)$ via
\begin{align}
    1+ \tan^2 \left(t\frac{\pi }{2}\right) \geq \left(\frac{\tan \left(t\frac{\pi }{2}\right)}{t \frac{\pi }{2}} \right)^2 \underset{(\ast \ast)}{\geq}B_{\ast \ast}(\epsilon) \geq B_{\ast}(\epsilon) .
\end{align}

For $0<t<1$, the first inequality in the preceding chain is strict.
This proves the strict determinant condition used in
Eq.~\eqref{eq:strict-hessian-at-contact}.
Together with the boundary cases discussed above, this proves the following
condition on the whole domain:
\begin{equation}
\begin{aligned}
F(\epsilon,t)=\widehat F(\epsilon,t)
\quad\Longleftrightarrow\quad
&2(1-\epsilon)\sin\left(\frac{\pi t}{2}\right)\\
&\geq\left[2+\sqrt{(2-\epsilon)\epsilon}-\epsilon(5-2\epsilon)\right]
\frac{\pi t}{2}\cos\left(\frac{\pi t}{2}\right).
\end{aligned}
\end{equation}
On the edges $\epsilon=0$ and $\epsilon=\epsilon_{max}$, the condition reduces
to $\sin u\geq u\cos u$ with $u=\pi t/2$, which holds throughout $[0,\pi/2]$.
On the edges $t=0$ and $t=1$, it is automatic.

 \subsubsection{The condition $(\ast\ast\ast)$ implies $(\ast\ast)$}

\label{sec: simple to conv}

To complete assertion $(ii)$ of Result~\ref{result app :main},
we show that $F(\epsilon,t)\leq t$ implies
$F(\epsilon,t)=\widehat F(\epsilon,t)$.
The boundary cases have already been established, so we assume
$0<\epsilon<\epsilon_{max}$ and $0<t<1$ in the following manipulations.
In this domain,

\begin{equation}\label{eq: cond main}
(\ast\ast\ast)\qquad  F\left(\epsilon,t\right)\leq t \quad \Longleftrightarrow \quad \frac{g(t)}{t}\leq \frac{1}{f(\epsilon)} \quad \Longleftrightarrow \quad \frac{\sin(t\frac{\pi}{2})}{t\frac{\pi}{2}}\leq \frac{1}{1 +\sqrt{2\epsilon-\epsilon^2} - \epsilon}.
\end{equation}
We now  show that this  inequality implies 
\begin{equation}
(\ast \ast\ast)\qquad \implies \qquad(\ast \ast) \qquad \frac{\tan \left(t\frac{\pi }{2}\right)}{t \frac{\pi }{2}} \geq \frac{2 +\sqrt{(2-\epsilon ) \epsilon } -\epsilon  (5-2 \epsilon ) }{2 (1-\epsilon )}.
\end{equation}
First, rewrite the these conditions as
\begin{align}
    &(\ast\ast\ast) \quad 1- \frac{\sin(t\frac{\pi}{2})}{t\frac{\pi}{2}}\geq 1-\frac{1}{1 +\sqrt{2\epsilon-\epsilon^2} - \epsilon} \\
    &(\ast \ast) \qquad \frac{\tan \left(t\frac{\pi }{2}\right)}{t \frac{\pi }{2}} -1\geq \frac{2 +\sqrt{(2-\epsilon ) \epsilon } -\epsilon  (5-2 \epsilon ) }{2 (1-\epsilon )}-1.
\end{align}
The implication $(\ast \ast\ast)\implies(\ast \ast)$ can therefore be guaranteed by the following chain of inequalities
\begin{equation}
    \frac{\tan \left(t\frac{\pi }{2}\right)}{t \frac{\pi }{2}} -1 \underset{(1)}{\geq}1- \frac{\sin(t\frac{\pi}{2})}{t\frac{\pi}{2}}\underset{(\ast \ast\ast)}{\geq} 1-\frac{1}{1 +\sqrt{2\epsilon-\epsilon^2} - \epsilon} \underset{(2)}{\geq}   \frac{2 +\sqrt{(2-\epsilon ) \epsilon } -\epsilon  (5-2 \epsilon ) }{2 (1-\epsilon )}-1,
\end{equation}
and it remains to demonstrate (1) and (2). 

For the first inequality we have
\begin{align}
    (1) \quad \Longleftrightarrow \quad  \tan \left(x \right)+\sin(x) \geq 2x \quad \text{for} \quad x\in[0,\pi/2).
\end{align}
Here, the functions on both sides of the inequality are equal to zero at $x=0$, and the inequality can be implied by showing that the derivatives are ordered on the whole interval. That is,
\begin{equation}
    \frac{\dd}{\dd x} (\tan \left(x \right)+\sin(x)) = \frac{1}{\cos^2(x)}+\cos(x)\geq 2 \quad \Longleftrightarrow \quad 1 + z^3 - 2 z^2 \geq 0 \quad \text{for} \quad z\in[0,1].
 \end{equation}
The last inequality can be verified by noting that $1 + z^3 - 2 z^2$ equals to zero at $z=1$ and is a non-increasing function of $z\in[0,1]$.

For the second inequality we have
\begin{align}
(2)\quad    1-\frac{1}{1 +\sqrt{2\epsilon-\epsilon^2} - \epsilon} \geq  \frac{2 +\sqrt{(2-\epsilon ) \epsilon } -\epsilon  (5-2 \epsilon ) }{2 (1-\epsilon )}-1  
\quad \Longleftrightarrow \quad \frac{(1+2\epsilon -2\epsilon^2) \left(\sqrt{(2-\epsilon ) \epsilon }-\epsilon \right)}{2\, (1-\epsilon ) \left(\sqrt{(2-\epsilon ) \epsilon }+1-\epsilon\right)}\geq 0.
\end{align}
Here, one can simply verify that all parenthesis contain terms which are nonnegative for $\epsilon\in[0,\epsilon_{max}]$. This concludes the proofs of the assertions $(i-iii)$ of the Result~\ref{result app :main}. 

\subsection{(iv) Tightness of the bounds   in the large $n$ limit.}
\label{app: tight RBT}

 We now construct an SRQ model saturating the bound of Eq.~\eqref{app: main} in the $n\to \infty$ limit. To do se let the source distribute the maximally entangled two qubit state $\ket{\Phi^+}=\frac{1}{\sqrt 2} (\ket{00}+\ket{11})$ and Alice perform the two measurements corresponding to the observables 
 \be
 A_x^\alpha =\left(\cos(\alpha) H +(-1)^x \sin(\alpha) M \right)
\ee
with $H=\frac{Z+X}{\sqrt 2}$, $M=\frac{Z-X}{\sqrt 2}$
and  some angle $\alpha \in [0,\frac{\pi}{4}]$. With this, let \BobS\, perform the the measurement maximizing the CHSH score in the SP, we have seen (e.g. in the derivation of Lemma \ref{lemma: CHSH}) that this results in the following score 
\be
\S =  2\sqrt{2} \left(\frac{\cos(\alpha)+\sin(\alpha)}{\sqrt 2}\right)= 2\sqrt{2}\,  C_\alpha .
\ee

To complete the SRQ model, in the LP after the router let Bob's system be measured by the  two-outcomes PVM measurement $\{ {\rm E}_\lambda\}$ in the eigenbasis of $H $. Doing so, and observing an output $\lambda=\pm 1$ with equal probability, collapses Alice's reduced state to 
\be
\varrho_{\pm}= \frac{1}{2}\left( \id \pm H \right).
\ee
A copy of the value $\lambda$ is sent to the measurement device of \BobL\, (manifestly there is no entanglement shared between Alice and \BobL) which is instructed to implement the deterministic response function $p(b|\lambda,\theta)\in \{0,1\}$. This is an LHS model, which conditional on \BobL's measurement output $b$ leaves Alice's qubit in the state
\be
\rho_{b|\theta} = \frac{1}{2} \Big( p(b|+,\theta)\, \varrho_+  +  p(b|-,\theta) \, \varrho_-\Big) = \frac{p(b|+,\theta) +p(b|-,\theta)}{4} \id + \frac{p(b|+,\theta) -p(b|-,\theta)}{4} H.
\ee
We now assume that the response function has the symmetry $p(0|\pm,\theta)=p(1|\mp,\theta)$, which allows us to write
\begin{align}
\rho_{0|\theta} &= \frac{p(0|+,\theta) +p(1|+,\theta)}{4} \id + \frac{p(0|+,\theta) -p(1|+,\theta)}{4} H
\\
\rho_{1|\theta} &= \frac{p(0|+,\theta) +p(1|+,\theta)}{4} \id - \frac{p(0|+,\theta) -p(1|+,\theta)}{4} H.
\end{align}

Accordingly, in the continuous setting limit $n\to \infty$ for our SRQ model the LP quantities are given by 
\begin{align}
\T&= \int_0^\pi \frac{\dd \theta}{\pi} \tr \left[\rho_{0|\theta}+\rho_{1|\theta}\right] \\
&= \int_0^\pi \frac{\dd \theta}{\pi} \big( p(0|+,\theta) +p(1|+,\theta)\big)\\
\C & = \int_0^\pi \frac{\dd \theta}{\pi} \tr \left[ (\cos(\theta) A_0 +\sin(\theta) A_1)(\rho_{0|\theta}-\rho_{1|\theta})\right] \\
& = \int_0^\pi \frac{\dd \theta}{\pi} \cos(\alpha) (\cos(\theta)+\sin(\theta)) \tr \left[  H (\rho_{0|\theta}-\rho_{1|\theta})\right]\\
& = \sqrt{2 }\cos(\alpha) \int_0^\pi \frac{\dd \theta}{\pi}\cos \left(\theta -\frac{\pi }{4}\right) \big( p(0|+,\theta) -p(1|+,\theta)\big),
\end{align}
where we used $\cos(\theta)+\sin(\theta) = \sqrt{2} \cos \left(\theta -\frac{\pi }{4}\right)$. Let us now chose a specific deterministic response function given by
\begin{equation}\label{eq:corrected-response}
\Bigl(p(0|+,\theta),p(1|+,\theta),p(\varnothing|+,\theta)\Bigr)
=\begin{cases}
(1,0,0) & \left|\theta-\frac{\pi}{4}\right|\leq\Omega,\\
(0,1,0) & \left|\theta-\frac{\pi}{4}-\pi\right|<\Omega,\\
(0,0,1) & \text{otherwise}.
\end{cases}
\end{equation}

for a parameter $\Omega\in\left[0,\frac{\pi}{2}\right]$.
We use a strict inequality in the second case to avoid overlap at
$\Omega=\pi/2$. This only changes the assignment at one value of $\theta$ and
does not affect the integrals.
The intervals of $\theta$ for which Bob detects a particle have total length $2\Omega$, and
\[
\int_0^\pi\cos\left(\theta-\frac{\pi}{4}\right)
\bigl[p(0|+,\theta)-p(1|+,\theta)\bigr]\,\dd\theta=2\sin\Omega.
\]
Straightforward integration gives
\begin{equation}
 \T = \frac{\Omega}{\frac{\pi}{2}} \qquad \text{and} \qquad \C = \sqrt{2 }\cos(\alpha) \frac{\sin(\Omega)}{\frac{\pi}{2}} = (C_\alpha + \sqrt{1-C_\alpha^2} )\frac{\sin(\Omega)}{\frac{\pi}{2}}.
\end{equation}

We have therefore shown that by varying the parameters $\alpha$ and $\Omega$ one can construct SRQ models that achieve the values
\be
(\S,\T,\C) = \left(\S, \T, \frac{\S +\sqrt{8-\S^2}}{2\sqrt{2}} \frac{2}{\pi}\sin(\T \frac{\pi}{2})\right)
\ee
for all $\T\in[0,1]$ and $\S \in[2,2\sqrt 2]$. Hence, in the limit on continuous settings $n\to \infty$, for all $\T=\T_\infty$ and $\S$ we have constructed an SRQ model satisfying
\be
\frac{\pi}{2} \cW_\infty = \frac{\pi}{2} \cW = \frac{\S +\sqrt{8-\S^2}}{2\sqrt{2}}  \sin(\T \frac{\pi}{2}) = {\rm R}(\T_\infty,\S) 
\ee
That is, an SRQ model saturating the bound
\be
n \sin(\frac{\pi}{2 n}) \cW_n \leq {\rm R}(\T_n,\S).
\ee
in the $n\to \infty$ limit where $n \sin(\frac{\pi}{2 n})\to \frac{\pi}{2}$.
By definition of the concave envelope, mixing such SRQ strategies allows one to saturate our routed Bell inequality
\be
n \sin(\frac{\pi}{2 n}) \cW_n \leq \widehat {\rm R}(\T_n,\S).
\ee
in the same limit; hence the bound it tight.

\section{Proof that the functions $h_\beta(x)$ in Eq.~\eqref{appeq:h} are nonincreasing and concave}\label{app:concavity}
Here we show that the family of functions
\begin{equation}
    h_{\beta}(x) := \sqrt{ f_{\beta}(x)^2-1}\,\, - \cos^{-1}({1}/{f_{\beta}(x)})\,,
\end{equation}
defined on the domain $x \in [1/\sqrt{2},1]$ are nonincreasing and concave, i.e., $h_{\beta}'(x) \le 0$ and $h_{\beta}''(x) \le 0$ whenever $\beta \in (0,1]$. The function $f_{\beta}(x)$ is given by
\begin{equation}
    f_{\beta}(x) := \frac{x + \sqrt{1-x^2}}{\beta}\,.
\end{equation}
Computing the first and second derivatives of $h_{\beta}(x)$, we find
\begin{align}
    h_{\beta}'(x) &= f_{ \beta}'(x) \frac{\sqrt{{ f}^2_{ \beta}(x)-1}}{{f}_{ \beta}(x)}\,,\qquad
    h_{\beta}''(x)= { f}_{ \beta}''(x) \frac{\sqrt{{f}^2_{\beta}(x)-1}}{{ f}_{ \beta}(x)}  +
    ({ f}_{ \beta}'(x) )^2 \frac{1}{{ f}^2_{ \beta}(x) \sqrt{{ f}^2_{ \beta}(x)-1}}\,,
\end{align}
where,
\begin{equation}
    { f}_{ \beta}'(x) = \frac{1}{\beta} \left(1 - \frac{x}{\sqrt{1-x^2}}\right) \le 0\,,\quad
    { f}_{ \beta}''(x) = \frac{1}{\beta} \left(\frac{-1}{(1-x^2)^{3/2}} \right) \le 0\,.
\end{equation}
The fact that $h_{ \beta}(x)$ is nonincreasing follows from ${f}_{ \beta}'(x) \leq 0$.  In turn,
$h_{\beta}(x)$ is concave if 
\be
h_{\beta}''(x) \leq 0 \quad \Longleftrightarrow \quad { f}_{ \beta}''(x) \,{ f}_{\beta}(x)\,({ f}^2_{ \beta}(x)-1) + ({ f}_{ \beta}'(x) )^2\leq 0 
\ee
Straightforward algebra gives
\begin{align}
    { f}_{ \beta}''(x) \,{ f}_{\beta}(x)\,({ f}^2_{ \beta}(x)-1) + ({ f}_{ \beta}'(x) )^2 =   -\frac{\beta^2 (x-x^3) + 3 \left( x-x^3\left(\frac{2+\beta^2}{3}\right)\right) +\left(2 (x^2-\beta^2)+1\right) \sqrt{1-x^2} }{\beta ^4 \left(1-x^2\right)^{3/2}},
\end{align}
which is indeed nonpositive since all terms in the numerator and denominator are nonnegative, implying  $h_{\beta}''(x)\leq 0$.

\end{document}